\documentclass[12pt, draftclsnofoot, onecolumn]{IEEEtran}
\usepackage{ifpdf}

\usepackage{cite}
\usepackage{hyperref}

\ifCLASSINFOpdf
   \usepackage[pdftex]{graphicx}
\else
   \usepackage[dvips]{graphicx}
\fi
\usepackage{color}
\usepackage{amsmath}
\usepackage{amsthm}
\usepackage{bm}
\newtheorem{thm}{Theorem}

\DeclareMathOperator*{\argmax}{arg\,max}
\DeclareMathOperator*{\argmin}{arg\,min}
\usepackage[linesnumbered,ruled,vlined]{algorithm2e}
\let\oldnl\nl
\newcommand{\nonl}{\renewcommand{\nl}{\let\nl\oldnl}}
\SetKwInput{KwInput}{Input}
\SetKwInput{KwOutput}{Output}

\usepackage{algorithmic}

\usepackage{array}
\usepackage{fixltx2e}

\usepackage{stfloats}
\usepackage{url}
\usepackage{amsmath} 
\usepackage{ragged2e}
\usepackage{subfig}
\usepackage{grffile}
\graphicspath{ {Images/} }

\usepackage{multirow}

\begin{document}
\title{\textcolor{black}{Energy-Efficient Task Offloading Under E2E Latency Constraints}}

\author{Mohsen~Tajallifar,~
	Sina~Ebrahimi,~
	Mohammad~Reza~Javan,
	Nader~Mokari,
	and~Luca~Chiaraviglio,
	\thanks{M. Tajallifar, S. Ebrahimi, and N. Mokari are with the Department
		of Electrical and Computer Engineering, Tarbiat Modares University, Tehran, 14115-111 Iran e-mail: nader.mokari@modares.ac.ir. M. Javan is with Shahrood University of Technology. L. Chiaraviglio is with University of Rome Tor Vergata.}}

\maketitle
\vspace{-50 pt}
\begin{abstract}
\textcolor{black}{In this paper, we propose a novel resource management scheme {that} jointly allocates the transmit power and computational resources in a centralized radio access network architecture. The network comprises a set of computing nodes to which the requested tasks of different users are offloaded. The optimization problem minimizes the energy consumption of task offloading while takes the end-to-end-latency, i.e., the transmission, execution, and propagation latencies of each task, into account. We aim to allocate the transmit power and computational resources such that the maximum acceptable latency of each task is satisfied. Since the optimization problem is non-convex, we divide it into two sub-problems, one for transmit power allocation and another for task placement and computational resource allocation. Transmit power is allocated via the convex-concave procedure. In addition, a heuristic algorithm is proposed to jointly manage computational resources and task placement. We also propose a feasibility analysis that finds a feasible subset of tasks. Furthermore, a disjoint method that separately allocates the transmit power and the computational resources is proposed as the baseline of comparison. A lower bound on the optimal solution of the optimization problem is also derived based on exhaustive search over task placement decisions and utilizing Karush–Kuhn–Tucker conditions. Simulation results show that the joint method outperforms the disjoint method in terms of acceptance ratio. Simulations also show that the optimality gap of the joint method is less than $5\%$.}
\end{abstract}

\begin{IEEEkeywords}
Mobile edge computing, task offloading, resource allocation, end-to-end latency, task placement.
\end{IEEEkeywords}

%
\IEEEpeerreviewmaketitle

\section{Introduction}\label{introduction}
\vspace{-5 pt}
\subsection{Background}
\IEEEPARstart{I}{n} {order to fulfill the requirements of 5G mobile networks, {key} enabling technologies such as network function virtualization (NFV) and multi-access/mobile edge computing (MEC) are introduced. With NFV, the network functions (NFs) that traditionally used dedicated hardware are implemented in applications running on top of commodity servers \cite{yi2018comprehensive}.}
{On the other hand, MEC aims to support {low-latency} mobile services by bringing the remote servers closer to the mobile users \cite{mach2017mobile,etsi2016mobile}.} Moreover, MEC enables the offloading of the computational burden of users' tasks to reduce the impact of the limited battery power of {user equipment (UE)}.  Note that when executing servers are NFV-enabled, they are able to process various types of tasks. As a result, there is no restriction on offloading a task to a predetermined server.

\begin{figure}[!t]
\centering
\includegraphics[width=9cm]{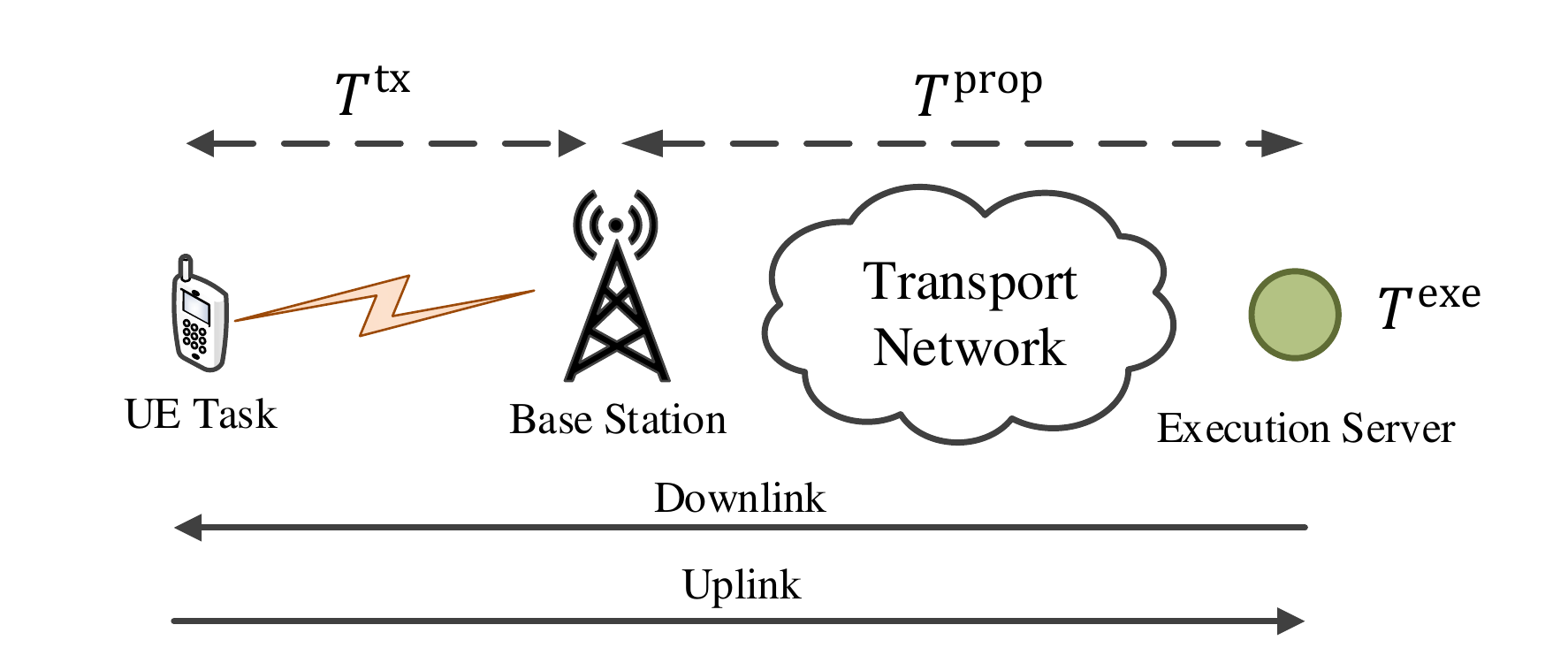}
\caption{\small{A typical task offloading example.}}
\label{offloading_fig}
\end{figure}
\setlength{\textfloatsep}{0pt}
A typical task offloading example is shown in Fig. \ref{offloading_fig}. In task offloading, the non-processed data of a task is sent from UE to an executing server that offloads the computational burden of the task execution on the executing server. As Fig. \ref{offloading_fig} shows, the user transmits the non-processed data of the task over the wireless link to its serving base station, which results in transmit latency $T^{\rm tx}$. Then, the received data is transmitted to an executing server. Executing servers are placed at the base station and distant nodes in the transport network. The data transmission through the transport network adds the propagation latency $T^{\rm prop}$ to the offloading process. Finally, the received data is processed at the executing server with execution latency $T^{\rm exe}$ and then is sent back to the user over the downlink. Therefore, the end-to-end (E2E) latency of task offloading is equal to the summation of $T^{\rm tx}$, $T^{\rm prop}$, and $T^{\rm exe}$ in both uplink and downlink.
\vspace{-15 pt}
\subsection{Related Works} \label{related-works}
\textcolor{black}{
We classify the related works on task offloading into four categories and discuss their applicability in practical scenarios.
\subsubsection{Task offloading to multiple executing servers}
In task offloading, a UE decides to whether offload a task to a single executing server or to select an executing server out of multiple servers. Offloading a task to one server in a set of executing servers in a multi-tier heterogeneous network is considered in \cite{guo2018computation, almughalles2019task}. Moreover, the authors in \cite{yang2018mobile, tran2019joint, dinh2017offloading} propose to offload a user's task to one of executing servers at base stations in a multi-cell network. Note that in the aforementioned works, the executing servers are located at the edge of the radio access network and the computational resources in the non-radio part of the network are not considered. In contrast, it is possible to offload a task to any server in the network in \cite{yang2018cost}, i.e., servers in radio access and non-radio parts of the network. However, radio resources are not allocated in \cite{yang2018cost}. Note that, ignoring computational resources in the non-radio part of the network or ignoring radio resource allocation results in an inefficient task offloading.
\subsubsection{Task placement and computational resource allocation}
Task offloading is comprised of two steps: i) task placement to select an executing server, and ii) computational resource allocation that allocates the resources of the executing server to each task. In this context, various works only focus on task placement with given computational resources for each task \cite{chen2015efficient, zhang2016energy, almughalles2019task, yang2018mobile, guo2018computation, li2017efficient, yang2018cost}, while others include resource allocation as well \cite{ zhao2017energy, zhang2017multi, zhang2017energy, wang2017joint, you2016energy, zhou2018computation, khalili2019joint, tran2019joint, zhang2018joint, xia2018power}. Note that the servers in non-radio part of the network are not involved in these works. As a result, computationally intensive tasks with moderate sensitivity to latency may occupy the capacity of executing servers in radio part of the network while high capacity servers in non-radio part of the network are underutilized.
\subsubsection{Joint Radio and Computational Resource Allocation}
Extensive research is made on joint radio and computational resource allocation \cite{xia2018power, zhang2018joint, tran2019joint, chen2018resource, li2018joint, khalili2019joint, zhang2016energy, zhou2018computation, you2016energy, guo2017energy, wang2017joint, chen2016joint, al2017energy, zhang2017energy, yu2016joint, zhang2017multi, zhao2017energy}. In these works, radio resources including transmit power and/or bandwidth as well as computational resources are allocated to each task. Energy-efficient resource allocation is performed in\cite{xia2018power, khalili2019joint,zhang2016energy, you2016energy, guo2017energy, al2017energy, yu2016joint, zhao2017energy}, and a weighted combination of consumed energy and latency is optimized in \cite{zhang2017multi, zhang2017energy, chen2016joint, wang2017joint, chen2018resource, tran2019joint, zhang2018joint}. Moreover, the impact of radio link quality without radio resource allocation is taken into account by \cite{yang2018mobile, almughalles2019task, chen2015efficient, dinh2017offloading, liu2016delay, li2018deep}. In these works, the latency of data transmission over radio links is taken into account, which impacts the optimal task placement. Note that although joint optimization of radio and computational resources increases the degrees of freedom in task offloading, the available computational resources in the radio access network are very limited, which limit the acceptance ratio of the network.}
\textcolor{black}{\subsubsection{Feasibility Analysis}
When task offloading is subjected to a maximum acceptable latency,  sufficient resources are required in various parts of the network. In case of insufficient resources, a feasibility analysis is needed to determine a feasible subset of requested tasks. One approach to face infeasibility is making some simplifying assumptions, e.g., assuming sufficient available resources for task offloading \cite{yang2018cost, chen2018resource} or offloading a task when it is beneficial, i.e., when offloading results in less energy consumption or latency \cite{chen2015efficient}. In practice, however, the resources are limited and tasks are subjected to execution deadlines. As a result, a feasibility analysis is inevitable. The feasibility analysis is performed by introducing a binary optimization variable, which is one when the task is accepted or zero when the task is rejected \cite{zhang2018joint, tran2019joint, khalili2019joint, zhang2017energy, zhao2017energy, xia2018power, li2017efficient, guo2018computation, yu2016joint}. Note that finding optimal binary variables results in combinatorial optimization problems that are challenging and of high complexity.}

\subsection{Motivation}
\textcolor{black}{The performance of a task offloading method is mainly measured by its latency and energy consumption. In practice, E2E latency comes from radio links, transport network links, and execution at the servers; and the energy consumption is impacted by consumed transmit power and computational resources.}

\textcolor{black}{Optimizing the performance of task offloading necessitates a joint optimization of all available resources in the network. However, existing works optimize a subset of resources and focus only on one part of the whole network. Moreover, the impact of E2E latency is not considered in the literature. As a result, existing methods may not perform well in practice.}

\textcolor{black}{In this paper, we propose a task offloading method that optimizes the energy consumption in terms of transmit power and computational resources under E2E latency constraints. Throughout the paper, the task offloading is referred to the process of transmit power allocation over radio links, task placement, i.e., selecting an executing server and its path, and computational resource allocation. The proposed method jointly allocates required transmit power to tasks, places each task in a proper NFV-enabled node, and allocates sufficient computational resources to the tasks. With this joint method, high latency of radio links caused by weak radio channels is compensated by a proper task placement and computational resource allocation. In contrast, high execution latency caused by limited computational resources is compensated by consuming more transmit power in radio links. As a result, more tasks are served, compared to a disjoint method wherein transmit power allocation is independent of task placement and computational resource allocation.}

\textcolor{black}{NFV enables a general-purpose server to execute various tasks without needing a specialized server for each task. Therefore, various tasks are dynamically offloaded to general-purpose executing servers in a network of NFV-enabled nodes instead of offloading each task to a respective specialized server. As a result, a task placement method is needed to determine an executing server and its route for each task. In spite of conventional routing methods that choose a route to a predetermined server, our task placement method jointly determines an executing server, the associated route to the executing server, and the required computational resources in the executing server.}

\textcolor{black}{We assume a deadline for offloading each task, i.e., sending the task from UE to the executing server and sending it back to UE performed under a maximum acceptable latency constraint. As a result, the sum of latencies in radio link, transport network links, and execution at the executing server is less than the maximum acceptable latency. The feasibility of this E2E offloading method depends on the available resources and location of executing servers in the network. For example, when the available transmit power is low, the radio link latency is large, which may violate E2E latency. In contrast, when the available computational resources at the executing server are low, the execution latency is large, which may also violate the E2E latency constraint. Therefore, our task offloading method includes a feasibility analysis that finds a set of feasible tasks.}

\textcolor{black}{The infeasibility of task offloading depends on the value of maximum acceptable latencies, i.e., lower values of maximum acceptable latencies result in a larger number of infeasible tasks and higher values result in a smaller number of infeasible tasks. Inspired by this fact and in contrast to the existing works, we add a non-negative variable to each maximum acceptable latency. Non-negative variables are zero for feasible tasks and are positive for infeasible tasks. Therefore, the set of feasible tasks is obtained by solving an optimization problem that minimizes the sum of non-negative variables, i.e., maximizes the number of feasible tasks.}

\textcolor{black}{Joint task offloading results in a non-convex problem due to coupling optimization variables. Moreover, the task placement is performed by obtaining binary variables, which makes the optimization problem further complicated. To deal with the optimization problem, we decouple transmit power allocation from task placement and computational resource allocation. Transmit power allocation is performed via the well-known convex-concave procedure (CCP) and a heuristic algorithm is proposed for task placement and computational resource allocation. CCP and the heuristic algorithm are alternatively applied until convergence. Note that both CCP and the heuristic algorithm preserve the monotonicity of convergence.}

\textcolor{black}{We also develop two baseline methods to evaluate the efficiency of our joint task offloading method. The first is a disjoint method in which transmit power allocation is performed independent of task placement. In doing so, the maximum acceptable E2E latency of each task is divided into a radio latency constraint and a non-radio latency constraint. We allocate transmit power under the radio latency constraint. Then, the task placement and computational resource allocation are performed under the non-radio latency constraint.}

\textcolor{black}{The second baseline method achieves a lower bound on the optimal solution of the joint task offloading optimization problem. The lower bound is achieved by relaxing some constraints in the optimization problem, which comes from leveraging practical assumptions such as orthogonality of wireless channels in large-scale antenna array systems. The optimal solution is then found by an exhaustive search over all feasible task placement candidates, finding the optimal computational resource allocation for each placement candidate, and choosing the placement candidate that results in the lowest objective value.}
\subsection{Contributions}
\textcolor{black}{In this paper, we develop an energy-efficient task offloading method that offloads the computational burden of a task from a UE to one of executing servers in a network of NFV-enabled nodes. In doing so, a task is offloaded by sending non-processed data of the task from the UE to a radio remote head (RRH) over a radio link, sending the data from the RRH toward the executing server through a transport network, and sending the processed data back from the executing server to UE. We assume that each task is offloaded under a respective deadline, i.e., the E2E latency of task offloading is less than the maximum acceptable latency of the task.}

\textcolor{black}{The main contributions and achievements of this paper are as follows:}
\begin{itemize}

	\item \textcolor{black}{We develop a joint task offloading method in a practical scenario, i.e., the proposed method allocates the transmit power, finds an executing server and the route to it, and allocates the computational resources in an energy-efficient manner. Moreover, the proposed method takes the E2E latency of task offloading into account. By the proposed method, the impact of weak radio links is compensated by placing the tasks in servers closer to UEs and consuming more computational resources. In contrast, limited computational resources are compensated by allocating more transmit power, resulting in an efficient and adaptive task offloading method.}

	\item\textcolor{black}{We propose a novel method for task placement and computational resource allocation. While the conventional routing methods find a route to a predetermined node, our proposed method jointly finds the executing server, its associated route, and the required computational resources in an energy-efficient manner.}

	\item\textcolor{black}{We find a lower bound on the objective function of the optimization problem in the feasibility analysis, i.e., an upper bound on the acceptance ratio of the proposed method. The lower bound is obtained by relaxing some of constraints in the optimization problem, performing an exhaustive search over all feasible task placement candidates, and finding the optimal computational resource allocation by utilizing Karush-Kuhn-Tucker conditions.}
	\item\textcolor{black}{Simulation results show that the proposed joint method outperforms its disjoint counterpart in terms of acceptance ratio. Moreover, the lower bound on the optimal solution is almost tight because the joint method attains the lower bound in practical scenarios. Specifically, the optimality gap of the joint method is less than $5\%$.}	
\end{itemize}
\subsection{Organization}
\textcolor{black}{The rest of the paper is organized as follows. Section \ref{system-model} introduces the system model. Section \ref{problem-formulation} describes the optimization problem formulation. In Section \ref{joint-ra}, we propose joint task offloading while disjoint task offloading and lower bound on optimal task offloading are proposed in Sections \ref{disjoint-ra} and \ref{opato_sec}, respectively. Simulation results are presented in Section \ref{simulation-results} and the paper is concluded in Section \ref{conclusion}.}
\subsection{{Notation}}
The {notation used in this paper} are given as follows. The vectors are denoted by bold lowercase symbols. Operators $\|\cdot\|$ and $|\cdot|$ are vector norm and absolute value of a scalar, respectively. $(\bf a)^{\rm T}$ is transpose of $\bf a$ and $[a]^+=\max(a,0)$. $\mathcal{A}\backslash\{a\}$ discards the element $a$ from the set $\mathcal{A}$. Finally, $\bf a\sim \mathcal{CN}({\bf 0},\Sigma)$ is a complex Gaussian vector with zero mean and covariance matrix $\bf \Sigma$.
\section{System Model}\label{system-model}
\textcolor{black}{The structure of the radio access network, channel model, and signaling scheme as well as NFV-enabled network, computational resources, and capacity of network links are described in this section.}
\vspace{-20 pt}
\subsection{Radio Access Network (RAN)}
We consider a centralized RAN architecture with a baseband unit (BBU) pool, which serves a set of $U$ RRHs, each equipped with $M$ antennas. {The} set of all users is denoted by $\mathcal K$. Each user is equipped with a single antenna and the total number of users is {$K=|\mathcal{K}|$}. The considered model is shown in Fig. \ref{sys_mod}. It is assumed that each RRH is connected to the BBU pool through a fronthaul link.\\
\indent We assume that each user requests a single task. Task $k$ is represented by a triplet $<L_k, D_k, T_k>$, where $L_k$ is the load of task $k$ (i.e., the required CPU cycles), $D_k$ is the data size of task $k$ (in terms of bits), and $T_k$ is the maximum acceptable latency of task $k$.\\
\begin{figure}[!t]
	\centering
	\includegraphics[width=6cm]{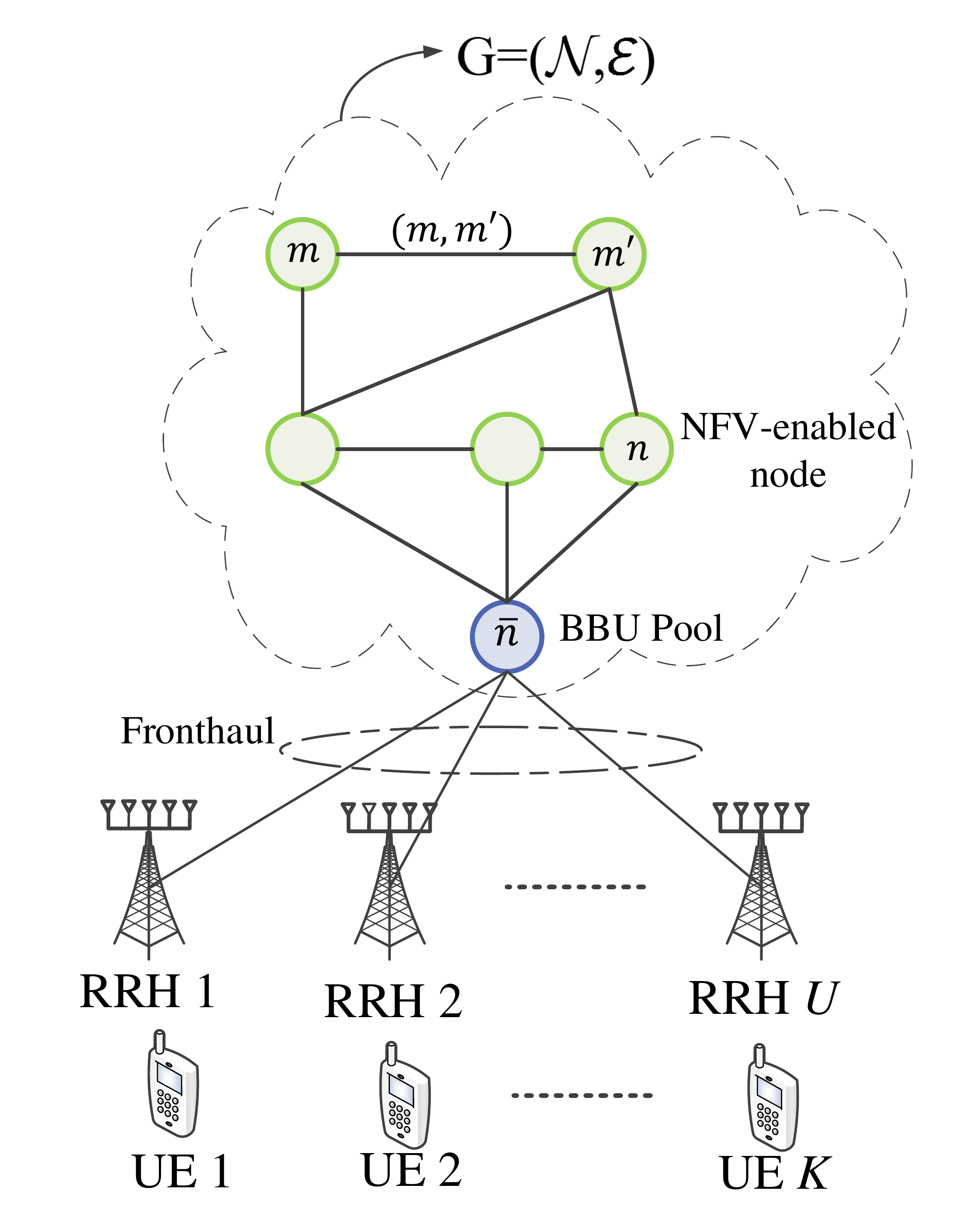}
	\caption{\small System model.}\label{sys_mod}
\end{figure}
\indent {Each UE transmits the non-processed data of its task to its serving RRH through a wireless link.}
We assume that each UE is served by a single RRH. The set of users served by RRH $u$ is $\mathcal{K}_u=\{k\in \mathcal{K}|J_u^k=1\}$ where $J_u^k$ is an indicator which equals 1 if UE $k$ is connected to RRH $u$ (0 otherwise). In this paper, we assume that the UE-RRH assignment is given and fixed. Focusing on the wireless link, we assume a narrow-band block fading channel model {\cite{xia2018power}}. The channel vector between UE $k$ and RRH $u$ is denoted by ${\bf h}_{u,k}$, where ${\bf h}_{u,k}=\sqrt{Q_{u,k}}\tilde{{\bf h}}_{u,k}$ in which $Q_{u,k}$ represents the path loss between RRH $u$ and UE $k$ and small-scale fading is modeled as $\tilde{{\bf h}}_{u,k}\sim \mathcal{CN}({\bf 0},{\bf I}_M)$. Similar to \cite{you2016energy, wang2017joint}, we assume that the channel state information (CSI) is constant over the offloading time. As we show through simulations, this assumption is non-restrictive in practical scenarios in sub-6 GHz bands.
 {UE $k$ transmits a symbol $x_k\sim \mathcal{CN}({\bf 0},1)$ with transmit power $\rho_k$ toward its serving RRH. The transmit power of UE $k$ is constrained to a maximum value, i.e., $\rho_k\le P^{\rm max}_k\quad \forall k$. The received signal vector at RRH $u$ is:
	\begin{equation}
	{\bf y}_u=\sum_{k\in \mathcal{K}}{\bf h}_{u,k}\sqrt{\rho_k}x_k, \quad \forall u.
	\end{equation}
	\textcolor{black}{We assume the maximum ratio combining (MRC) at RRHs because of its simplicity. Nevertheless, MRC is asymptotically optimal in massive MIMO systems \cite{lu2014overview}. Therefore, the combined signal is:}
	\begin{equation}
	{\bf z}_u={\bf F}_u^{\rm H} {\bf y}_u, \quad \forall u,
	\end{equation}
	where ${\bf F}_u=[{\bf f}_k], \forall k\in \mathcal{K}_u$ and ${\bf f}_k= \frac{{\bf h}_{u,k}}{\|{\bf h}_{u,k}\|}$. The estimated signal of UE $k$ is:
	\begin{align}
	\nonumber z_{k}=& {\bf f}_{k}^{\rm H}{\bf h}_{u,k}\sqrt{\rho_k}x_k+\sum_{j\in \mathcal{K}\backslash\{k\}}{\bf f}_{k}^{\rm H}{\bf h}_{u,j}\sqrt{\rho_k}x_k+ {\bf f}_{k}^{\rm H}{\bf n}_u, \quad \forall k\in \mathcal{K}_u,
	\end{align}
	where ${\bf n}_u\sim \mathcal{CN}({\bf 0},\sigma_{\rm n}^2{\bf I}_M)$ is the received noise vector at RRH $u$. Thus,} the signal to interference plus noise ratio (SINR) of UE $k$ is:
	\begin{equation}
	{\rm SINR}_k=\frac{\|{\bf h}_{u,k}\|^2 \rho_k}{\sum_{j\in \mathcal{K}\backslash\{k\}}\frac{|{\bf h}_{u,k}^{\rm H}{\bf h}_{u,j}|^2}{\|{\bf h}_{u,k}\|^2}\rho_j+\sigma_{\rm n}^2},\quad \forall k\in \mathcal{K}_u.
	\end{equation}
\textcolor{black}{Hence, the achievable data rate by UE $k$ is $R_{k}=W\log_2(1+{\rm SINR}_{k})\footnote{For wide-band channel model, the data rate of UE $k$ is the sum rate over all sub-carriers allocated to UE k.}$
bits per second (bps), where $W$ is the radio access network bandwidth. The radio transmission latency of task $k$ {in the uplink} is $T^{\rm tx}_k=\frac{D_k}{R_k}$\footnote{No buffering is assumed in the transport network routing. Therefore, transmission time of tasks' data over the transport network links is not taken into account.}. The sum of data rates of UEs served by RRH $u$ is less than the capacity of its fronthaul link, i.e., $\sum_{k\in \mathcal{K}_u}{R_k}\le B_{{\rm f},u}, \forall u$.
In this paper, similar to \cite{li2018deep,chen2015efficient}, and \cite{zhang2016energy}, we assume that the processed data size of task $k$ is small. Moreover, since the power budget of RRHs is generally large, the radio transmission latency in the downlink is assumed negligible.}
\subsection{NFV-enabled Network}
\textcolor{black}{The NFV-enabled network includes a graph $G=(\mathcal N, \mathcal E)$, where $\mathcal N$ and $\mathcal E$ are the set of nodes and edges (or links), respectively. A typical node in $\mathcal N$ is denoted by $n$ while the BBU pool is indicated by $\bar n$ (which is also a {node in} $\mathcal N$). The link between two nodes $m$ and $m'$ is denoted by $(m,m')$.
Each NFV-enabled node is comprised of an executing server and a routing device. The processing capacity (i.e., the maximum CPU cycles per second that are carried out) of the executing server in NFV-enabled node $n$  is indicated by $\Upsilon_n$. Moreover, the capacity of link $(m,m')$ is indicated by $B_{(m,m')}$ in terms of bps.}

\textcolor{black}{In this paper, we assume the full offloading scheme, i.e., the task of each user is completely executed in an executing server in the NFV-enabled network. Therefore, there is a need for placing each task to a proper executing server. A task placement decision consists of selecting an NFV-enabled node $n$ and its associated path from $\bar n$ to $n$ . We denote the $b^{\rm th}$ path between nodes $\bar{n}$ and $n$ as $p_{n}^{b}$ where $b\in\mathcal{B}_n=\{1 \cdots B_n\}$ and $B_n$ is the total number of paths between nodes $\bar{n}$ and $n$. \textcolor{black}{Note that a path between $\bar{n}$ and $n$ may comprise some intermediate nodes, which only forward the tasks' data via their routing devices and do not deliver the data to their executing servers.} We define decision variable $\xi_{p_{n}^{b}}^k$, which equals $1$ when task $k$ is offloaded to node $n$ and sent over path $p_{n}^{b}$ ($0$ otherwise). Each task is offloaded to one and only one node and path when we have:
\begin{equation}
\sum_{n\in\mathcal N}\sum_{b\in\mathcal{B}_n} \xi_{p_{n}^{b}}^k=1,\quad \forall k.
\end{equation}
Indicator $I^{(m,m')}_{p_{n}^{b}}$ determines whether a link contributes to a path. The indicator is equal to $1$ when link $(m,m')$ contributes to path $p_{n}^{b}$ {(0 otherwise)}.
Moreover, the set of all links that contribute to path $p_{n}^{b}$ is {$\mathcal{E}_{p_{n}^{b}}=\left\{(m,m')\in \mathcal{E}|I^{(m,m')}_{p_{n}^{b}}=1\right\}$}.
The amount of computational resources allocated to task $k$ is denoted by $\upsilon_k$ (in terms of CPU cycles per second). Note that the execution of each tasks is performed at only one node. To ensure that the allocated computational resources do not violate the processing capacity of that node, we should have{:}
\begin{equation}\label{node_constraint}
\sum_{k\in \mathcal{K}} \sum_{b\in \mathcal{B}_n}\upsilon_k \xi_{p_{n}^{b}}^k\le \Upsilon_n, \quad \forall n.
\end{equation}
Since the data of task $k$ is sent over the network with rate $R_k$, the aggregated rates of all tasks that pass a link should not exceed its capacity, which is guaranteed by the following constraint:
\begin{equation}\label{link_constraint}
\sum_{k\in \mathcal{K}}\sum_{n\in\mathcal N}\sum_{b\in \mathcal{B}_n}I^{(m,m')}_{p_{n}^{b}} \xi_{p_{n}^{b}}^k R_k\le B_{(m,m')}, \quad \forall (m,m')\in \mathcal E.
\end{equation}
The execution latency of task $k$ is $T^{\rm exe}_k=\frac{L_k}{\upsilon_k}$. The processed data of task $k$ is sent toward the BBU pool (i.e., node $\bar{n}$). In this paper, we assume the path of uplink and downlink are the same. Therefore, the overall propagation latency of task $k$ over the path $p_n^b$ is twice the propagation latency of path $p_n^b$. Thus, the propagation latency of task $k$ is {$T^{\rm prop}_k=2\sum_{ n\in\mathcal N} \sum_{b\in \mathcal{B}_n}\sum_{(m,m')\in \mathcal{E}_{p_{n}^{b}}} \xi_{p_{n}^{b}}^k \delta_{(m,m')}$}, where $\delta_{(m,m')}$ is the propagation latency of link $(m,m')$. Table \ref{table-notations} summarizes the notation used in the paper.}
\begin{table}[!h]
	\scriptsize
	\renewcommand{\arraystretch}{1.05}
	\centering
	\caption{\small Main Notation.}
	\label{table-notations}
	\begin{tabular}{| c| c| c|c|}			
		\hline
		\textbf{Notation}& \textbf{Definition} & \textbf{Notation} & \textbf{Definition}\\\hline		
		$U, M, K$&Number of RRHs, antennas and users & $W$ &Radio access network bandwidth\\ \hline
		$\mathcal{K, N, E}$ &Set of all users, nodes and links & $\mathcal{K}_u$ &Set of users served by RRH $u$\\ \hline
		$P^{\rm max}_k$ & {Power budget of UE $k$} & ${\bf h}_{u,k}$ &Channel vector between user $k$ and RRH $u$\\ \hline
		$L_k, D_k, T_k$ &\begin{tabular}{c} Load, data size and maximum\\acceptable latency of task $k$\end{tabular} & $\xi_{p_{n}^{b}}^k$ & \begin{tabular}{c} Decision variable for assignment of node $n$\\and its associated path $p_{n}^{b}$ to task $k$\end{tabular}\\ \hline
		$\Upsilon_n$ &Processing capacity of node $n$ & $B_{{\rm f},u}$ & Capacity of fronthaul link for RRH $u$\\ \hline
		\begin{tabular}{c}$B_{(m,m')},$\\ $\delta_{(m,m')}$ \end{tabular} &\begin{tabular}{c} Capacity and propagation latency\\of link $(m,m')$\end{tabular} & $\Lambda_n$ &\begin{tabular}{c}Computational energy efficiency coefficient\\of the node $n$\end{tabular}\\ \hline
		$p_{n}^{b}$ &$b^{\rm th}$ path between nodes $\bar{n}$ and $n$ & $\upsilon_k$ & Computational resources allocated to task $k$\\ \hline
		$\mathcal{B}_n$ &Set of paths between nodes $\bar{n}$ and $n$ & $\rho_k$ & Allocated transmit power to UE $k$\\ \hline
		$\mathcal{E}_{p_{n}^{b}}$ &Set of all links that contribute in path $p_{n}^{b}$ & $\alpha_k$ & Non-negative variable of task $k$\\ \hline
		$I^{(m,m')}_{p_{n}^{b}}$ & {\begin{tabular}{c} Indicator determining whether link $(m,m')$\\contributes in path $p_{n}^{b}$\end{tabular}} & $R_k$ &Data rate of task $k$ \\ \hline
		$J_u^k$ & \begin{tabular}{c} Indicator determining whether UE $k$\\is assigned to RRH $u$\end{tabular} & {$T^{\rm exe}_k$} & {Execution latency of task $k$}\\\hline
		{$T^{\rm tx}_k$} & {Radio transmission latency of task $k$}  & {$T^{\rm prop}_k$} & {Propagation latency of task $k$}\\\hline
	\end{tabular}
\end{table}
\section{Problem Formulation}\label{problem-formulation}
\textcolor{black}{In this section, we formulate the optimization problem of joint task offloading. Each task is offloaded under its E2E latency constraint and in an energy-efficient manner. The objective function is $\mathcal{E}({\boldsymbol \xi},{\boldsymbol \upsilon},{\boldsymbol \rho})= \sum_{k\in \mathcal{K}}\rho_k+\eta\sum_{n\in \mathcal N}\sum_{k\in \mathcal{K}}\sum_{b\in\mathcal{B}_n} \Lambda_n \xi_{p_{n}^{b}}^k {\upsilon_k}^3$,
where ${\boldsymbol \xi}=[\xi_{p_{1}^{1}}^1,\cdots,\xi_{p_{N}^{B_N}}^K]^{\rm T},{\boldsymbol \upsilon}=[\upsilon_1,\cdots,\upsilon_K]^{\rm T}$, and {${\boldsymbol \rho}=[\rho_1,\cdots,\rho_K]^{\rm T}$} are the vectors of all $\xi_{p_{n}^{b}}^k,\upsilon_k,$ and $\rho_k$, respectively; $\Lambda_n$ denotes the computational energy efficiency coefficient of node $n$ \cite{zhou2018computation}, and $\eta$ is a weight. Note that the first term in $\mathcal{E}$ is the transmit power consumption and the second term is the power consumption of executing servers. Therefore, the joint task offloading optimization problem is:
\begin{equation}\label{joint_op}
\begin{array}{ll}
\mathop{\min}\limits_{{\boldsymbol \xi},{\boldsymbol \upsilon},{\boldsymbol \rho}}& \mathcal{E}({\boldsymbol \xi},{\boldsymbol \upsilon},{\boldsymbol \rho})\\
\text{s.t.}&  \text{C1:}\quad T^{\rm exe}_k+T^{\rm prop}_k+T^{\rm tx}_k\le T_k, \quad \forall k,\\
&\text{C2:}\quad {\sum_{k\in \mathcal{K}} \sum_{b\in \mathcal{B}_n}\upsilon_k \xi_{p_{n}^{b}}^k\le \Upsilon_n, \quad \forall n,}\\
&\text{C3:}\quad {\sum_{k\in \mathcal{K}}\sum_{n\in\mathcal N}\sum_{b\in \mathcal{B}_n}I^{(m,m')}_{p_{n}^{b}} \xi_{p_{n}^{b}}^k R_k\le B_{(m,m')}, \quad \forall (m,m')\in \mathcal E,}\\
&\text{C4:}\quad {\sum_{k\in \mathcal{K}_u}{R_k}\le B_{{\rm f},u},\quad \forall u,}\\
&\text{C5:}\quad \rho_{k}\le P^{\rm max}_k,\quad \forall k,\\
&\text{C6:}\quad \sum_{n\in\mathcal N}\sum_{b\in\mathcal{B}_n} \xi_{p_{n}^{b}}^k=1,\quad \forall k,\\
\end{array}
\end{equation}
{under variables: ${\boldsymbol \xi}\in \{0,1\},{\boldsymbol \upsilon}\ge 0,{\boldsymbol \rho}\ge 0$.}
Constraint C1 guarantees that the maximum acceptable latency of task offloading is respected. Constraints C2 and C3 make sure that all tasks are offloaded without violation in processing capacity of nodes and capacity of links, respectively. Constraint C4 {ensures the} capacity of fronthaul links. Constraint C5 guarantees the power budget of UEs while constraint C6 makes sure that each task is offloaded to only one node and path.}

\section{Joint Task Offloading (JTO)}\label{joint-ra}
In this section, we solve optimization problem \eqref{joint_op}. \textcolor{black}{This problem is non-convex due to integer variable $\boldsymbol \xi$ and coupling variables in C1-C4. Therefore, we solve \eqref{joint_op} by decoupling transmit power allocation from task placement and computational resource allocation. In doing so, transmit power is allocated given task placement and allocated computational resources. Then, we perform task placement and computational resource allocation having allocated transmit powers. The proposed approach needs a feasible initialization. However, it is likely for constraint C1 to make \eqref{joint_op} infeasible. Thus, we need to propose a feasibility analysis to find a feasible subset of tasks.}
\vspace{-14pt}
\subsection{Feasibility Analysis}\label{feasibility_section}
\textcolor{black}{The feasible set of \eqref{joint_op} is extended by adding a non-negative variable $\alpha_k$ to the maximum acceptable latency of task $k$. Thus, the feasibility problem is constructed by replacing the objective function of \eqref{joint_op} with the sum of non-negative variables, i.e., $\sum_{k=1}^K \alpha_k$ \cite{Chinneck2008feasibility}. The constraints which cause infeasibility are found by solving the feasibility problem and determining the constraints with positive values of non-negative variables.} The feasibility problem is:
\begin{equation}\label{joint_elas}
\begin{array}{ll}
\mathop{\min}\limits_{{\boldsymbol \xi},{\boldsymbol \upsilon},{\boldsymbol \rho},{\boldsymbol \alpha}}& \sum_{k\in \mathcal{K}}\alpha_k\\
\text{s.t.}&  \text{C1-a:}\quad T^{\rm exe}_k+T^{\rm prop}_k+T^{\rm tx}_k \le T_k+\alpha_k, \quad \forall k\in\mathcal{K}\\
& \text{C2-{C6},}
\end{array}
\end{equation}
under variables: ${\boldsymbol \xi}\in \{0,1\},{\boldsymbol \upsilon}\ge 0,{\boldsymbol \rho}\ge 0,{\boldsymbol \alpha}\ge 0$. \textcolor{black}{Note that non-negative variables are added only to C1 because when C1 is eliminated, the optimization problem (\ref{joint_op}) is always feasible. Thus, we seek for the tasks whose maximum acceptable latencies are violated and eliminate them one by one until a subset of feasible tasks remains. The solution to (\ref{joint_elas}) not only provides the infeasible constraints but also determines the level of infeasibility, i.e., constraints with larger values of non-negative variables need more resources to become feasible. Therefore, we first eliminate the tasks with larger values of non-negative variables.}

Without loss of equivalence, we add the summation of inequalities in C1-a as a new constraint C7. Therefore, optimization problem \eqref{joint_elas} is restated as:
\begin{equation}\label{joint_elas_sum}
\begin{array}{ll}
\mathop{\min}\limits_{{{\boldsymbol \xi},{\boldsymbol \upsilon},{\boldsymbol \rho},\boldsymbol \alpha}}& \sum_{k\in \mathcal{K}}\alpha_k\\
\text{s.t.}&  \text{C1-a:}\quad T^{\rm exe}_k+T^{\rm prop}_k+T^{\rm tx}_k \le T_k + \alpha_k, \quad \forall k\\
& \text{C2-{C6},}\\
& \text{{C7}:}\sum_{k\in \mathcal{K}}\left(T^{\rm exe}_k+T^{\rm prop}_k+T^{\rm tx}_k -T_k\right) \le \sum_{k\in \mathcal{K}} \alpha_k. \quad
\end{array}
\end{equation}
\noindent This optimization problem is equivalent with:

\begin{equation}\label{joint_elas_obj}
\begin{array}{ll}
\mathop{\min}\limits_{{{\boldsymbol \xi},{\boldsymbol \upsilon},{\boldsymbol \rho},\boldsymbol \alpha}}& \sum_{k\in \mathcal{K}}\left(T^{\rm exe}_k+T^{\rm prop}_k+T^{\rm tx}_k \right)\\
\text{s.t.}&  \text{C1-a, and C2-C6,}\\
\end{array}
\end{equation}
in which  the term $\sum_{k\in\mathcal{K}}T_k$ is removed from the objective because it is constant. \textcolor{black}{We solve \eqref{joint_elas_obj} by decoupling transmit power allocation from task placement and coumputational resource allocation. In other words, we solve \eqref{joint_elas_obj} under variables $\boldsymbol \upsilon,\boldsymbol \xi,\boldsymbol \alpha$, having $\boldsymbol \rho$ fixed and vice versa. To perform task placement and computational resource allocation, we need an initial $\boldsymbol \rho = {\boldsymbol \rho}^0$ that satisfies C3 and C4, which are satisfied with a small value of $R_k$, i.e., small values of $\rho_k$. Next, we solve the following optimization problem:
\begin{equation}\label{asm}
\begin{array}{ll}
\mathop{\min}\limits_{{\boldsymbol \alpha},{\boldsymbol \upsilon},{\boldsymbol \xi}}& \sum_{k\in \mathcal{K}}\left(T^{\rm exe}_k+T^{\rm prop}_k \right)\\
\text{s.t.}&  \text{C1-a, C2,C3, and C6}\\
\end{array}
\end{equation}
by a heuristic method. As in Algorithm \ref{heuristic}, we find variables $\boldsymbol \xi$ and $\boldsymbol \upsilon$ that minimize the objective of \eqref{asm}. Then, we set the non-negative variables for a feasible C1. In doing so, for task $k$, we calculate the amount of unused computational resources at all nodes, formally expressed as $\tilde{\Upsilon}_n^k=\Upsilon_n-\sum_{ j\in \mathcal{K}\backslash \{k\}} \sum_{b\in\mathcal{B}_{n}}\upsilon_j \xi_{p_{n}^{b}}^j$. Morevoer, the available capacity of link $(m,m')$ is $\tilde{B}_{(m,m')}^k=B_{(m,m')}-\sum_{j\in\mathcal{K}\backslash \{k\}}\sum_{n\in\mathcal N}\sum_{b\in\mathcal{B}_{n}}I^{(m,m')}_{p_{n}^b}  \xi_{p_{n}^{b}}^jR_j$. A task is placed in node $n$ only when there is a feasible path between $\bar{n}$ and $n$, i.e., a path with sufficient capacity in all of its links. The set of all such nodes is $\mathcal{N}_k$. For each $n \in \mathcal{N}_k$, we calculate $T_k^{\rm exe} + T_k^{\rm prop}$ when $\upsilon_k = \tilde{\Upsilon}_n^k$. Next, we find the node and feasible path that give the smallest $T_k^{\rm exe} + T_k^{\rm prop}$, denoted by $n^\star$ and $b^\star$, respectively.
Note that from C1, the sufficient computational resources allocated to task $k$ is $\upsilon_{\rm temp} = \frac{L_k}{T_k-T_k^{\rm tx}-T_k^{\rm prop}}$. When $\tilde{\Upsilon}_{n^\star}^k \ge \upsilon_{\rm temp}$, C1 is satisfied by setting $\upsilon_k = \upsilon_{\rm temp}$ and $\alpha_k = 0$. Otherwise, we set $\upsilon_k = \tilde{\Upsilon}_{n^\star}^k$ and $\alpha_k = T_k^{\rm tx} + T_k^{\rm exe} + T_k^{\rm prop}-T_k$. Next, the available computational resources of nodes and available capacity of links are updated and this process is repeated for all of tasks. Note that we begin with tasks that require lower resources, i.e., tasks with lower values of $T_k$.}
\begin{algorithm}
	\DontPrintSemicolon
	\KwInput{${\boldsymbol \rho}$}
	\textbf{{sort $\boldsymbol \alpha$:} $T_{[1]}\le T_{[2]}\le \cdots T_{[|\mathcal{K}|]}$}\\
	\For{$k=[1]:[|\mathcal{K}|]$}
	{
		\nonl \% Find a feasible node according to capacity of paths terminated at that node \\
		$\mathcal{N}_k=\{n\in \mathcal{N}| \exists b: R_k\le \tilde{B}_{(m,m')}^k \forall (m,m')\in \mathcal{E}_{p_{n}^b}\}$\\
		$\tilde{\Upsilon}_n^k=\Upsilon_n-\sum_{ j\in \mathcal{K}\backslash \{k\}} \sum_{b\in\mathcal{B}_{n}}\upsilon_j \xi_{p_{n}^{b}}^j,\quad \forall n$\\
		\nonl \%Find the best node and its associated path\\
		$(n^\star,{b^\star})=\argmin \limits_{n\in\mathcal{N}^k,b\in\mathcal{B}_n} T^{\rm exe}(\tilde{\Upsilon}_n^k)+T^{\rm prop}({p_n^b})$\\
		set $\quad \xi_{p_{n^\star}^{b^\star}}^{k} = 1$ and $\xi_{p_{n}^{b}}^{k} = 0, \forall (n,b) \neq (n^\star,{b^\star})$\\
		\nonl \% Update computational resource allocation and non-negative variables\\
		$\upsilon_{\rm temp} = \frac{L_k}{T_k-T_k^{\rm tx}-T_k^{\rm prop}}$\\
		\If{$\tilde{\Upsilon}_{n^\star}^k \ge \upsilon_{\rm temp}$}
		{
			$\text{set}\quad \upsilon_k^\star=\upsilon_{\rm temp}$ and $ \alpha_k^\star=0$
		}
		\Else
		{
			$\upsilon_k^\star=\tilde{\Upsilon}_{n^\star}^k $ and $\alpha_k^\star = T_k^{\rm tx} + T_k^{\rm exe} + T_k^{\rm prop}({p_{n^\star}^{b^\star}})-T_k$
		}		
	}
	\KwOutput{${\boldsymbol \alpha^\star,\boldsymbol\xi^\star,\boldsymbol \upsilon^\star
		}$}
	\caption{Heuristic Algorithm for Solving \eqref{asm}.}
	\label{heuristic}
\end{algorithm}

\textcolor{black}{After solving \eqref{asm}, we allocate the transmit power by solving:
\begin{equation}\label{elas_p}
\begin{array}{ll}
\mathop{\min}\limits_{{\boldsymbol \rho}}& \sum_{k=1}^{K}T_k^{\rm tx}\\
\text{s.t.}&\text{C1-a, and C3-C5}.
\end{array}
\end{equation}
Note that in the heuristic method, we have $T_k^{\rm tx} + T_k^{\rm exe} + T_k^{\rm prop} = T_k + \alpha_k $. As a result, any feasible solution to \eqref{elas_p} does not increase $T_k^{\rm tx}$ because \eqref{elas_p} is infeasible for larger values of $T_k^{\rm tx}$. Hence, replacing \eqref{elas_p} with its feasibility problem counterpart does not impact the decreasing monotonicity of the objective function in \eqref{joint_elas_obj}. The feasibility problem of \eqref{elas_p} is:
\begin{equation}\label{elas_p0}
\begin{array}{ll}
\text{find}& \boldsymbol \rho\\
\text{s.t.}&\text{C1-a, and C3-{C5}}.
\end{array}
\end{equation}
In solving \eqref{elas_p0}, we note that the constraints C1-a, C3  and C4 are non-convex. Therefore, we need to find a convexified version of \eqref{elas_p0}.We use CCP \cite{boyd2016ccp} to convexify (\ref{elas_p0}). In doing so, we reformulate C1-a as:
\begin{equation}\label{concave_R_k}
R_k \ge \frac{D_k}{T_k+\alpha_k-T^{{\rm prop},i}_k-T^{{\rm exe},i}_k}.
\end{equation}
where $T^{{\rm exe},i}_k$ and $T^{{\rm prop},i}_k$ are the execution latency and propagation latency of task $k$ obtained from the heuristic method in $i^{\rm th}$ iteration, respectively. In order to convexify (\ref{concave_R_k}), we need a concave approximation of $R_k$ with respect to $\boldsymbol \rho$. The rate $R_k$ is:
\begin{equation}
\label{R_k_first}
R_k=W\log_2\bigg(\frac{\sum_{j\in \mathcal{K}}\frac{|{\bf h}_{u,k}^{\rm H}{\bf h}_{u,j}|^2}{|{\bf h}_{u,j}|^2}\rho_j+\sigma^2_{\rm n}}{\sum_{j\in \mathcal{K}\backslash \{k\}}\frac{|{\bf h}_{u,k}^{\rm H}{\bf h}_{u,j}|^2}{|{\bf h}_{u,j}|^2}\rho_j+\sigma^2_{\rm n}}\bigg), \quad k \in \mathcal{K}_u,
\end{equation}
{which is equivalent to:}
\begin{equation}\label{R_k}
\hspace{-5pt}
R_k=\underbrace{W\log_2\bigg(\sum_{u=1}^{U}\sum_{j\in \mathcal{K}_{u}}\frac{|{\bf h}_{u,k}^{\rm H}{\bf h}_{u,j}|^2}{|{\bf h}_{u,j}|^2}\rho_j+\sigma^2_{\rm n}\bigg)}_{h_k({\boldsymbol \rho})}-\underbrace{W\log_2\bigg(\sum_{u=1}^{U}\sum_{j\in \mathcal{K}_{u}\backslash\{k\}}\frac{|{\bf h}_{u,k}^{\rm H}{\bf h}_{u,j}|^2}{|{\bf h}_{u,j}|^2}\rho_j+\sigma^2_{\rm n}\bigg)}_{g_k({\boldsymbol \rho})}.
\end{equation}
Both $h_k({\boldsymbol \rho})$ and $g_k({\boldsymbol \rho})$ are concave functions of $\boldsymbol \rho$. Thus, we need to find a linear approximation of $g_k({\boldsymbol \rho})$, which is
$\hat{g}_k({\boldsymbol \rho})= g_k({\boldsymbol \rho}^0)+\nabla g_k({\boldsymbol \rho}^0)^{\rm T} ({\boldsymbol \rho}-{\boldsymbol \rho}^0)$, where:
\begin{equation}
[\nabla g_k({\boldsymbol \rho})]_i=\left\{
\begin{array}{ll}
\frac{W\sum_{u=1}^UI_u^i\frac{|{\bf h}_{u,k}^{\rm H}{\bf h}_{u,i}|^2}{|{\bf h}_{u,i}|^2}}{\ln(2)\left(\sum_{u=1}^{U}\sum_{j\in \mathcal{K}_{u}\backslash\{k\}}\frac{|{\bf h}_{u,k}^{\rm H}{\bf h}_{u,j}|^2}{|{\bf h}_{u,j}|^2}\rho_j+\sigma^2_{\rm n}\right)}, i\in \mathcal{K}\backslash\{k\},\\
0, \qquad i=k.
\end{array}
\right.
\end{equation}
Next, we focus on the convex approximation of C3 and C4. To this aim, we find a convex approximation of $R_k$, which is found by linear approximation of $h_k({\boldsymbol \rho})$.
Thus, we have
$\hat{h}_k({\boldsymbol \rho})= h_k({\boldsymbol \rho}^0)+\nabla h_k({\boldsymbol \rho}^0)^{\rm T} ({\boldsymbol \rho}-{\boldsymbol \rho}^0)$, where:
\begin{align}
[\nabla h_k({\boldsymbol \rho})]_i=& \frac{W\sum_{u=1}^{U}I_u^i\frac{|{\bf h}_{u,k}^{\rm H}{\bf h}_{u,i}|^2}{|{\bf h}_{u,i}|^2}}{\ln(2)\left(\sum_{u=1}^{U}\sum_{j\in \mathcal{K}_{u}}\frac{|{\bf h}_{u,k}^{\rm H}{\bf h}_{u,j}|^2}{|{\bf h}_{u,j}|^2}\rho_j+\sigma^2_{\rm n}\right)}, \quad i\in \mathcal{K},
\end{align}
Finally, the convexified version of \eqref{elas_p} is:
\begin{equation}\label{elas_p_convex}
\begin{array}{ll}
\text{find}& \boldsymbol \rho\\
\text{s.t.}& \text{C1-b} \quad h_k({\boldsymbol \rho})- \hat{g}_k({\boldsymbol \rho})\ge \frac{D_k}{T_k+\alpha_k-T^{{\rm prop},i}_k-T^{{\rm exe},i}_k}, \quad \forall k\in \mathcal{K}\\
&\text{C3-a:}\quad\sum_{k\in \mathcal{K}}\sum_{n\in\mathcal N}\sum_{b\in\mathcal{B}_n}I^{(m,m')}_{p_{n}^{b}} \xi_{p_{n}^b}^k \left(\hat{h}_k({\boldsymbol \rho})- g_k({\boldsymbol \rho})\right)\le B_{(m,m')}, \quad \forall (m,m')\in \mathcal E\\
&\text{C4-a:} \quad\sum_{k\in \mathcal{K}_u}{\left(\hat{h}_k({\boldsymbol \rho})- g_k({\boldsymbol \rho})\right)}\le B_{{\rm f},u}, \forall u\\
&\text{C5:} \quad \rho_k\le P^{\rm max}_k, \forall k,
\end{array}
\end{equation}
under variable: $\boldsymbol\rho \ge 0$. Note that, based on CCP, any feasible solution of \eqref{elas_p_convex} is also feasible in \eqref{elas_p0} \cite{boyd2016ccp}. The feasibility problem (\ref{joint_elas}) is solved by alternatively solving \eqref{asm} and \eqref{elas_p_convex}. Then, we reject the task that makes \eqref{joint_op} infeasible. According to Algorithm \ref{feasiblity_alg}, we find the value of the maximum non-negative variable. If the value is positive, its associated task is rejected, the set of served tasks is updated, and \eqref{joint_elas} is solved for updated set of tasks. This procedure continues until all non-negative variables are zero. The output of Algorithm \ref{feasiblity_alg} is feasible subset of tasks $\mathcal{K}^\star$ as well as the solution of \eqref{joint_elas}, i.e., the values of ${\boldsymbol \xi}^{\rm ini}, {\boldsymbol \rho}^{\rm ini}, \text{and } {\boldsymbol \upsilon}^{\rm ini}$, which are utilized as initialization for solving (\ref{joint_op}).}
\begin{algorithm}
	\DontPrintSemicolon
	{\nonl {\textbf{Initialize:}{$\quad \mathcal{K}=\{1,\cdots,K\}, \quad \boldsymbol \xi ={\bf 0}, \quad{\boldsymbol \rho}^0:$ very small}}}\\
	\Repeat{$\sum_{k\in \mathcal{K}}\alpha_k=0$}
	{
	 $i = 0$\\
		\Repeat{$\sum_{k\in \mathcal{K}}\alpha_k^{i}-\sum_{k\in \mathcal{K}}\alpha_k^{i+1}\le \epsilon$ or $i\ge I_{\rm max}$}
		{
			{\nonl {\% Allocate transmit power, computational resources, and place the tasks}}\\
			Solve \eqref{asm} via Algorithm \ref{heuristic} and return ${\boldsymbol \upsilon}^{i+1}$, ${\boldsymbol \xi}^{i+1}$, and ${\boldsymbol \alpha}^{i+1}$\\
			Solve (\ref{elas_p_convex}) and return ${\boldsymbol \rho}^{i+1}$\\
			$i=i+1$\\
		}
		\nonl \% Discard the infeasible task \\	
		$k^\star=\argmax\limits_{k\in \mathcal{K}} \alpha_k$\\
		\If{$\alpha_{k^\star}> 0$}
		{
			$\mathcal{K}=\mathcal{K}\backslash \{k^\star\}$
		}	
	}
	\KwOutput{${\boldsymbol \xi}^{\rm ini}={\boldsymbol \xi}^{i+1}, {\boldsymbol \rho}^{\rm ini}={\boldsymbol \rho}^{i+1}, {\boldsymbol \upsilon}^{\rm ini}={\boldsymbol \upsilon}^{i+1}$, and $\mathcal{K}^\star = \mathcal{K}$}
	\caption{JTO Feasibility Analysis for Solving \eqref{joint_elas}.}
	\label{feasiblity_alg}
\end{algorithm}
\vspace{-5pt}
\subsection{Optimization}
\textcolor{black}{Given the feasible solution ${\boldsymbol \xi}^{\rm ini}, {\boldsymbol \rho}^{\rm ini}, {\boldsymbol \upsilon}^{\rm ini}$, and the set of accepted tasks $\mathcal{K}^{\star}$, we seek for the solution of (\ref{joint_op}). Similar to Algorithm \ref{feasiblity_alg}, we decouple power allocation from task placement and coumputational resource allocation. The optimization problemof task placement and coumputational resource allocation is:
\begin{equation}\label{op_c}
\begin{array}{ll}
\mathop{\min}\limits_{{\boldsymbol \upsilon},{\boldsymbol \xi}}& \sum_{n\in \mathcal N}\sum_{k\in \mathcal{K}}\sum_{b\in\mathcal{B}_n} \Lambda_n \xi_{p_{n}^b}^{k} {\upsilon_k}^3\\
\text{s.t.}&  \text{C1-C3, and C6},\\
\end{array}
\end{equation}
which is non-convex. Note that the objective of \eqref{op_c} is an increasing function of $\upsilon_k$ and allocating lower computational resources to task $k$ decreases the power consumption. But, allocating lower computational resources increases execution latency and violates the E2E latency constraint. As a result, we need to find nodes with smaller propagation latency to compensate for increased execution latency. In doing so, we find a subset of nodes with smaller propagation latency than the current executing server and with sufficient capacity of links terminating at that nodes. This set of nodes is $\mathcal{N}'_k=\{n'\in \mathcal{N}| \exists b': R_k\le \tilde{B}_{(m,m')}^k \forall (m,m')\in \mathcal{E}_{p_{n'}^{b'}} \text{ and } T_{k}^{\rm prop}(p_{n'}^{b'}) \le T_{k}^{\rm prop}(p_{n}^{b})\}$, where we assume task $k$ is previously placed through path $p_{n}^{b}$. For each node in $\mathcal{N}'_k$, we calculate the minimum computational resources that satisfy the E2E latency constraint, i.e., $\upsilon_{\rm temp} = \frac{L_k}{T_k-T_k^{\rm tx}-T_k^{\rm prop}(p_{n'}^{b'})}$. When $\tilde{\Upsilon}_{n'}^k \ge \upsilon_{\rm temp}$ and $\Lambda_{n'}\upsilon_{\rm temp}^3 \le \Lambda_{n}\upsilon_k^3$, we ensure that task placement through $p_{n'}^{b'}$ and computational resource allocation $\upsilon_{\rm temp}$ are feasible and result in lower power consumption. Therefore, we set $\upsilon_k = \upsilon_{\rm temp}$. Otherwise, we reinstate $\upsilon_k$ for task $k$. Algorithm \ref{heuristic_op} begins with the tasks with larger power consumption, i.e., $\Lambda_{n_k}\upsilon_k^3$, where $n_k$ denotes the executing server of task $k$. This procedure is repeated for all accepted tasks.}
\begin{algorithm}
	\DontPrintSemicolon
	\KwInput{${\boldsymbol \xi}^{\rm ini}, {\boldsymbol \rho}^{\rm ini}, {\boldsymbol \upsilon}^{\rm ini}$}
	\textbf{sort:} $\Lambda_{[1]}\upsilon_{[1]}^3\le \Lambda_{[2]}\upsilon_{[2]}^3 \le \cdots \Lambda_{[K]}\upsilon_{[K]}^3$\\
	\For{$k=[1]:[|\mathcal{K}|]$}
	{
		\nonl \% Find a feasible node according to capacity of paths terminated at that node \\
		$\mathcal{N}'_k=\{n'\in \mathcal{N}| \exists b': R_k\le \tilde{B}_{(m,m')}^k, \forall (m,m')\in \mathcal{E}_{p_{n'}^{b'}} \text{ and } T_{k}^{\rm prop}(p_{n'}^{b'}) \le T_{k}^{\rm prop}(p_{n}^{b})\}$\\
		\For{$ n' \in \mathcal{N}'_k$}
	{
		$\upsilon_{\rm temp} = \frac{L_k}{T_k-T_k^{\rm tx}-T_k^{\rm prop}(p_{n'}^{b'})}$\\
		$\tilde{\Upsilon}_{n'}^k=\Upsilon_{n'}-\sum_{ j\in \mathcal{K}\backslash \{k\}} \sum_{b\in\mathcal{B}_{n'}}\upsilon_j \xi_{p_{n'}^{b}}^j$\\
			\If{$ \upsilon_{\rm temp} \ge \tilde{\Upsilon}_{n'}^k$ {\rm and } $\Lambda_{n'}\upsilon_{\rm temp}^3 \le \Lambda_{n^\star}\upsilon_k^3$}
		{
			set $\quad \upsilon_k^\star=\upsilon_{\rm temp}$\\
			set $\quad \xi_{p_{n'}^{b'}}^{k\star} = 1$ and $\xi_{p_{n}^{b}}^{k\star} = 0, \forall (n,b) \neq (n',{b'})$\\
			\bf{break}
		}	
	}
	}
	\KwOutput{${\boldsymbol\xi^\star,\boldsymbol \upsilon^\star}$}
	\caption{Heuristic Algorithm for Solving \eqref{op_c}.}
	\label{heuristic_op}
\end{algorithm}

\vspace{-15 pt}
The sub-problem of transmit power allocation, after convexification, is:
\begin{equation}\label{op_p}
\begin{array}{ll}
\mathop{\min}\limits_{{\boldsymbol \rho}}& \sum_{k\in \mathcal{K}}\rho_k\\
\text{s.t.}& \text{C1-c:} \quad h_k({\boldsymbol \rho})- \hat{g}_k({\boldsymbol \rho})\ge \frac{D_k}{T_k-T^{{\rm prop},i}_k-T^{{\rm exe},i}_k}, \forall k\in \mathcal{K}\\
&\text{C3-a, C4-a, and C5.}\\
\end{array}
\end{equation}
Based on CCP in Algorithm \ref{ccp_alg} and starting from ${\boldsymbol \rho}^0 = {\boldsymbol \rho}^{\rm ini}$, an iterative solution of (\ref{op_p}) provides a sub-optimal transmit power allocation.
\begin{algorithm}

	\DontPrintSemicolon	
	\KwInput{ ${\boldsymbol \rho}^0 = {\boldsymbol \rho}^{\rm ini}$, $i=0$, $\epsilon=10^{-3}$, $I_{\rm max}^{\rm \rho}=10^2$}	
	\Repeat{$\sum_{k\in \mathcal{K}}\rho_k^{i}-\sum_{k\in \mathcal{K}}\rho_k^{i+1}\le \epsilon$ \rm{ or } $i\ge I_{\rm max}^{\rm \rho}$}
	{
		\nonl \% Allocate power to users\\
		
		Solve (\ref{op_p}) and return ${\boldsymbol \rho}^{i+1}$	\\
		$i=i+1$\\
	}
	\KwOutput{${\boldsymbol \rho}^{\star} = {\boldsymbol \rho}^{i+1}$}
	\caption{Power Allocation in JTO.}
	\label{ccp_alg}
\end{algorithm}
Finally, optimization problem \eqref{joint_op} is solved via Algorithm \ref{optimization_alg}, which alternatively solves optimization problem \eqref{asm} via Algorithm \ref{heuristic_op} and optimization problem \eqref{op_p} via Algorithm \ref{ccp_alg}.

{\textcolor{black}{From the implementation point of view, BBU is responsible for gathering the required information, performing resource allocation, and sending the decisions to the associated entities. Specifically, in JTO, BBU needs to acquire CSI of UEs and the available computational resources in the NFV-enabled nodes. CSI of each UE is estimated at its serving RRH and is forwarded through fronthaul links with negligible latency. In addition, each NFV-enabled node sends the available computational resources to the BBU through the transport network. After performing JTO, BBU transmits the value of allocated powers to RRHs. Next, BBU forwards the received data of tasks as well as the obtained computational resources to associated NFV-enabled nodes based on task placement variables. In the downlink, the processed data of tasks are sent to BBU, which in turn transmits UEs processed data to their serving RRH.}
\vspace{-15 pt}
\begin{algorithm}
	\DontPrintSemicolon	
	\KwInput{${\boldsymbol \xi}^0={\boldsymbol \xi}^{\rm ini}, {\boldsymbol \rho}^0={\boldsymbol \rho}^{\rm ini}, {\boldsymbol \upsilon}^0={\boldsymbol \upsilon}^{\rm ini}, \mathcal{K}^\star$, $i=0$}
	\Repeat{$\mathcal{E}({\boldsymbol \xi}^{i},{\boldsymbol \upsilon}^{i},{\boldsymbol \rho}^{i})-\mathcal{E}({\boldsymbol \xi}^{i+1},{\boldsymbol \upsilon}^{i+1},{\boldsymbol \rho}^{i+1})\le \epsilon$ or $i\ge I_{\rm max}$}
	{
		\nonl \% Place the tasks and allocate the computational resources\\
		Solve \eqref{op_c} via Algorithm \ref{heuristic_op} and return ${\boldsymbol \upsilon}^{i+1}$ and ${\boldsymbol \xi}^{i+1}$\\
		\nonl \% Allocate the transmit power\\
		Solve (\ref{op_p}) via CCP in Algorithm \ref{ccp_alg} and return ${\boldsymbol \rho}^{i+1}$\\
		$i=i+1$\\	
	}
	\KwOutput{${\boldsymbol \xi}^\star, {\boldsymbol \rho}^\star, {\boldsymbol \upsilon}^\star$}
	\caption{JTO Optimization Algorithm for Solving \eqref{joint_op}.}
	\label{optimization_alg}
\end{algorithm}
\vspace{-35 pt}
\subsection{Convergence analysis}
In this subsection, we prove the convergence of Algorithms \ref{feasiblity_alg} and \ref{optimization_alg}.
\begin{thm}
Algorithm \ref{feasiblity_alg} is convergent.
\end{thm}
\begin{proof}
\textcolor{black}{We show that the objective value of \eqref{joint_elas}, i.e., $\sum_{k\in \mathcal{K}}\alpha_k$, is non-increasing in each step of Algorithm \ref{feasiblity_alg} and since the objective value is lower bounded by zero, Algorithm \ref{feasiblity_alg} is convergent. In $i^{\rm th}$ iteration of Algorithm \ref{feasiblity_alg}, Algorithm \ref{heuristic} sets $\alpha_k^{i+1}$ either equal to $0$ when E2E latency of task $k$ is guaranteed or equal to $T_k^{\rm tx} + T_k^{\rm exe} + T_k^{\rm prop}-T_k$ when E2E latency is larger than its maximum acceptable value. Therefore, we have $\alpha_k^{i+1}= [T^{\rm exe}_k +T^{\rm prop}_k +T^{\rm tx}_k -T_k]^+$.  Hence, we need to show that $\sum_{k\in \mathcal{K}}(T^{\rm tx}_k+T^{\rm exe}_k+T^{\rm prop}_k)$ does not increase after $i^{\rm th}$ iteration. Algorithm \ref{heuristic} affloads task $k$ so that $T^{\rm exe}_k+T^{\rm prop}_k$ in the objective of \eqref{asm} is minimized (line 5  in Algorithm \ref{heuristic}). As a result, Algorithm \ref{heuristic} does not increase the objective value of \eqref{asm}, i.e., $\sum_{k\in\mathcal{K}}(T^{\rm prop}_k (\boldsymbol\xi^{i+1})+T^{\rm exe}_k (\boldsymbol\upsilon^{i+1}))\le \sum_{k\in\mathcal{K}}(T^{\rm prop}_k (\boldsymbol\xi^i)+T^{\rm exe}_k (\boldsymbol\upsilon^i))$. Moreover, as discussed in subsection \ref{feasibility_section}, Algorithm \ref{heuristic} makes C1-a active, i.e., $T^{\rm tx}_k (\boldsymbol\rho^i)=T_k+ \alpha_k^{i+1}-T^{\rm exe}_k (\boldsymbol\upsilon^{i+1})-T^{\rm prop}_k (\boldsymbol\xi^{i+1})$, and therefore, any feasible solution to \eqref{elas_p0} does not increase the objective vlaue of \eqref{elas_p}, i.e.,  $\sum_{k\in \mathcal{K}}T^{\rm tx}_k(\boldsymbol\rho^{i+1})\le \sum_{k\in \mathcal{K}}T^{\rm tx}_k(\boldsymbol\rho^{i})$, which gives $\sum_{k\in\mathcal{K}}(T^{\rm exe}_k (\boldsymbol\upsilon^{i+1})+T^{\rm prop}_k (\boldsymbol\xi^{i+1})+T^{\rm tx}_k (\boldsymbol\rho^{i+1})) \le \sum_{k\in\mathcal{K}}(T^{\rm exe}_k (\boldsymbol\upsilon^i)+T^{\rm prop}_k (\boldsymbol\xi^i)+T^{\rm tx}_k (\boldsymbol\rho^i))$. As a result, we have $\sum_{k\in \mathcal{K}}\alpha_k^{i+1}\le \sum_{k\in \mathcal{K}}\alpha_k^i$, that is, Algorithm \ref{feasiblity_alg} is convergent.}

\textcolor{black}{Note that Algorithm \ref{feasiblity_alg} may eliminate the task with maximum non-negative variable. This elimination is equivalent to removing the constraints of \eqref{joint_elas} associated with the eliminated task. Note that, eliminating a task increases the available capacity of links in transport network and available computational resources in NFV-enabled nodes. As a result, a search space of Algorithm \ref{heuristic} increases, which may result in lower propagation and execution latencies. Moreover, eliminating a task extends the feasible set of \eqref{elas_p0}. Therefore, data rate of users may increase, which in turn may decrease $\sum_{k\in \mathcal{K}}T^{\rm tx}_k$. As a result, eliminating the task with maximum non-negative variable does not increase the objective of \eqref{joint_elas}.}
\end{proof}
\begin{thm}
Algorithm \ref{optimization_alg} is convergent.
\end{thm}
\begin{proof}
\textcolor{black}{Algorithm \ref{optimization_alg} solves \eqref{joint_op} by alternating minimization of \eqref{op_c} and \eqref{op_p}. Therefore, we need to show that Algorithm \ref{heuristic_op} (which solves \eqref{op_c}) and Algorithm \ref{ccp_alg} (which solves \eqref{op_p}) do not increase the objective value of \eqref{joint_op}. According to line 7 of Algorithm \ref{heuristic_op}, computational resource allocation and task placement do not increase the objective value of \eqref{op_c}. In addition, based on \cite{boyd2016ccp}, convergence of Algorithm \ref{ccp_alg} is guaranteed  and CCP does not increase the objective of \eqref{op_p}. As a result, the objective value of \eqref{joint_op} is non-increasing in each iteration, and since ${\Psi}({\boldsymbol \xi}, {\boldsymbol \upsilon}, {\boldsymbol \rho})$ is lower bounded by zero, Algorithm \ref{optimization_alg} is convergent.}
\end{proof}
\vspace{-20 pt}
\subsection{Summary of JTO}
\textcolor{black}{Herein we summarize JTO. We obtain a set of feasible tasks by solving \eqref{joint_elas}. In doing so, we decouple the power allocation from task placement and computational resource allocation, which are performed by solving \eqref{elas_p0} and Algorithm \eqref{heuristic}, respectively. Then, we solve \eqref{joint_op} for feasible tasks via Algorithm \ref{optimization_alg}, which includes the alternating minimization of \eqref{op_c} and \eqref{op_p} via Algorithm \ref{heuristic_op} and Algorithm \ref{ccp_alg}, respectively.}
\vspace{-10 pt}
\subsection{Computational Complexity (CC) Analysis} \label{computational-complexity}
\textcolor{black}{In this section, we analyze CC of the proposed algorithms. CC of JTO is equal to CC of feasibility analysis in Algorithm \ref{feasiblity_alg} and CC of optimization in Algorithm \ref{optimization_alg}. Algorithm \ref{feasiblity_alg} includes two nested while loops and CC of the inner loop is equal to CC of Algorithm \ref{heuristic} and CC of solving \eqref{elas_p_convex}. CC of Algorithm \ref{heuristic} depends on the required computations for calculation of the parameters in Algorithm \ref{heuristic}, which are provided in Table \ref{compute_ac} where $B$ is the maximum number of the paths between any node and $\bar{n}$ whereas $E$ is the total number of the edges in network graph $G$. Hence, CC of Algorithm \ref{heuristic} is {$CC^1 = \mathcal{O} (K^2 N B  E)$}.
    \begin{table}[h]
        \centering
        \caption{\small CC of obtaining parameters in Algorithm \ref{heuristic}.}
	\label{compute_ac}
	\begin{tabular}{ c c c c}
		\hline
		\textbf{Parameter}& \textbf{CC}&\textbf{Parameter}& \textbf{CC}\\\hline
		$\mathcal{N}^k$ & $N\times E$ & $\tilde{B}_{(m,m')}^k$ & $N \times B \times (K-1)$ \\ \hline
		$\tilde{\Upsilon}_n^k$ &$B\times (K-1)$ & $(n^\star, b^\star)$ & $N \times B $ \\ \hline
	\end{tabular}
    \end{table}
}

\textcolor{black}{
Problem \eqref{elas_p_convex} is solved via CVX, which exploits IPM for finding the optimal solution \cite{cvx}. Based on \cite{mokari2015limited} and \cite{boyd2004convex}, the required number of iterations for IPM to converge is $\frac{\log N_{\rm c}}{t^0 \varrho\log \varsigma}$, where $N_{\rm c} = 2K+U+E$ is the total number of constraints in \eqref{elas_p_convex}, $t^0$ is the initial point for approximation of the barrier function, $\varrho$ is the desired accuracy of convergence and $0<\varsigma\ll 1$ is used for updating the stepsize of the barrier function accuracy. Note that the inner loop in Algorithm \ref{feasiblity_alg} is repeated at most $K I_{\rm max}$ times. As a result, CC of Algorithm \ref{feasiblity_alg} is
\begin{equation}
CC^{2}=K\times I_{\rm max}\left(\frac{\log (2K+U+E)}{t^0 \varrho\log \varsigma}+\mathcal{O} (K^2 N B E)\right).
\end{equation}
Algorithm \ref{optimization_alg} includes alternating minimization of \eqref{op_c} and \eqref{op_p} via Algorithm \ref{heuristic_op} and Algorithm \ref{ccp_alg}, respectively. Note that CC of Algorithm \ref{heuristic_op} is in the same order as that of Algorithm \ref{heuristic}, i.e., $CC^3 = CC^1$. Algorithm \ref{ccp_alg} solves \eqref{op_p} at most $I_{\rm max}^{\rho}$ times and CC of solving \eqref{op_p} is equal to CC of \eqref{elas_p_convex}. As a result, CC of Algorithm \ref{ccp_alg} is
\begin{equation}
CC^4 = I_{\rm max}^{\rho} \frac{\log (2K+U+E)}{t^0 \varrho\log \varsigma}.
\end{equation}
Algorithms \ref{heuristic_op} and \ref{ccp_alg} are executed at most $I_{\rm max}$ times. Therefore, CC of Algorithm \ref{optimization_alg} is $CC^5 = I_{\rm max}\left(CC^3 + CC^4\right)$. Finally, CC of JTO is $CC^{\rm JTO} = CC^2 +CC^5$. Note that  Algorithm \ref{heuristic}, Algorithm \ref{heuristic_op}, and IPM are of polynomial time complexity. Therfore, JTO is also of polynomial time complexity.
}

\section{Disjoint Task Offloading (DTO)} \label{disjoint-ra}
\textcolor{black}{In DTO, transmit power allocation is independent of task placement and computational resource allocation. The transmit power is allocated under a radio latency constraint, i.e., $T_k^{\rm tx}\le T_k^{\rm RAN}$. Then, the task placement and computational resource allocation are performed so that $T_k^{\rm prop} + T_k^{\rm exe} \le T_k - T_k^{\rm RAN}$. The convexified sub-problem of the transmit power allocation is:}
\begin{equation}\label{pow_op_disjoint}
\begin{array}{ll}
\mathop{\min}\limits_{{\boldsymbol \rho}}& \sum_{k\in \mathcal{K}}\rho_k\\
\text{s.t.}& \text{C1-d:} \quad h_k({\boldsymbol \rho})- \hat{g}_k({\boldsymbol \rho})\ge \frac{D_k}{T_k^{\rm RAN}}, \quad \forall k\in \mathcal{K}\\
&\text{C4-a, and C5}.
\end{array}
\end{equation}
According to discussion on discussion on \eqref{joint_op}, a feasibility analysis is needed for \eqref{pow_op_disjoint}. Similar to JTO, the feasibility problem of (\ref{pow_op_disjoint}) is:
\begin{equation}\label{pow_elas_disjoint}
\begin{array}{ll}
\text{find} &\boldsymbol \rho\\
\text{s.t.}& \text{C1-e:} \quad h_k({\boldsymbol \rho})- \hat{g}_k({\boldsymbol \rho})\ge \frac{D_k}{T_k^{\rm RAN}+\alpha_k}, \quad \forall k\in \mathcal{K}\\
&\text{C4-a, and C5},\\
\end{array}
\end{equation}
which is solved via CVX. Next, the non-negative variables are updated as $\alpha_k=[T^{\rm tx}_k-T^{\rm RAN}_k]^+$ and the task with maximum non-negative variable is eliminated. This procedure is repeated until a feasible subset of tasks for transmit power allocation is obtained. After this step, \eqref{pow_op_disjoint} is solved with the feasible subset of tasks. The transmit power allocation phase of DTO is provided in Algorithm \ref{pow_alg_disjoint}.
\begin{algorithm}
	\DontPrintSemicolon
		\KwInput
	{
		$\mathcal{K}=\{1,\cdots,K\}, {\boldsymbol \alpha}^0:$ very large,
		${\boldsymbol \rho}^0:$ very small, $T_k^{\rm RAN}=(0,T_k)$
	}
	\Repeat{$\sum_{k\in \mathcal{K}}\alpha_k=0$}
	{
	$i = 0$\\
		\Repeat{$\sum_{k\in \mathcal{K}}\alpha_k^{i}-\sum_{k\in \mathcal{K}}\alpha_k^{i+1}\le \epsilon$ or $i\ge I_{\rm max}$}
		{			
			\nonl \% Allocate the transmit power to users \\
			Solve (\ref{pow_elas_disjoint}) via CVX and set ${\boldsymbol \rho}^{i+1}={\boldsymbol \rho}^\star$\\
			\nonl \% Update the non-negative variables\\
			$\alpha^{i+1}_k=[T_k^{\rm tx} - T_k^{\rm RAN}]^+,\quad \forall k \in \mathcal{K}$\\
			$i=i+1$
		}
		$k^\star=\argmax_{k\in \mathcal{K}} \alpha_k $\\
	\nonl \% Discard the infeasible task\\
		\If{$\alpha_{k^\star}> 0$}
		{
			$\mathcal{K}=\mathcal{K}\backslash \{k^\star\}$
		}	
	}
	\nonl \% Minimize the transmit power\\
	Solve (\ref{pow_op_disjoint}) via CCP in Algorithm \ref{ccp_alg} and return ${\boldsymbol \rho}^\star$ \\
	\KwOutput{${\boldsymbol \rho}^{\rm \star}, \mathcal{K}^{\rm RAN} = \mathcal{K}$}
	\caption{DTO Transmit Power Allocation.}
	\label{pow_alg_disjoint}
\end{algorithm}

\textcolor{black}{Having obtained transmit power $\boldsymbol \rho$, task placement and computational resource allocation are performed, whose associated sub-problem is:}
\begin{equation}\label{op_c_disjoint}
\begin{array}{ll}
\mathop{\min}\limits_{{\boldsymbol \xi},{\boldsymbol \upsilon}}& \sum_{n\in \mathcal N}\sum_{k\in \mathcal{K}}\sum_{b\in\mathcal{B}_n} \Lambda_n \xi_{p_{n}^b}^{k} {\upsilon_k}^3\\
\text{s.t.}& \text{C1-f:} \quad T_k^{\rm prop} + T_k^{\rm exe} \le  T_k- T_k^{\rm RAN}, \quad \forall k\in \mathcal{K},\\
&  \text{C2, C3, and C6.}
\end{array}
\end{equation}
\textcolor{black}{A feasibility analysis is also needed for solving \eqref{op_c_disjoint}. Similar to the transmit power allocation, we introduce a set of non-negative variables. The resulting sub-problem is similar to \eqref{asm} by replacing C1-a with C1-f, which is solved by algorithm \ref{heuristic}. After obtaining a set of feasible tasks, \eqref{op_c_disjoint} is solved via Algorithm \ref{heuristic_op}. The feasibility analysis and optimization of DTO is provided in Algorithm \ref{task_offloading_alg_disjoint}.}
\begin{algorithm}
	\DontPrintSemicolon	
	\KwInput
	{
		$\mathcal{K}^{\rm RAN}$, ${\boldsymbol \xi} = \bf 0$\\
	}
	\Repeat{$\sum_{k\in \mathcal{K}}\alpha_k=0$}
	{
		$i = 0$\\
		\Repeat{$\sum_{k\in \mathcal{K}}\alpha_k^{i}-\sum_{k\in \mathcal{K}}\alpha_k^{i+1}\le \epsilon$ or $i\ge I_{\rm max}$}
		{
			\nonl \% Allocate transmit power, computational resources, and place the tasks\\
			Solve \eqref{op_c_disjoint} via Algorithm \ref{heuristic} given ${\boldsymbol \upsilon}^{i}$, ${\boldsymbol \xi}^{i}$, and ${\boldsymbol \alpha}^{i}$ and return ${\boldsymbol \upsilon}^{i+1}$, ${\boldsymbol \xi}^{i+1}$, and ${\boldsymbol \alpha}^{i+1}$\\
			$i=i+1$\\
		}
	\nonl \% Find the task with maximum non-negative variable\\
		$k^\star=\argmax_{k\in \mathcal{K}} \alpha_k$\\
		\If{$\alpha_{k^\star}> 0$}
		{
			\nonl \% Discard the infeasible task\\
			$\mathcal{K}=\mathcal{K}\backslash \{k^\star\}$
		}	
	}
	$i=0$\\
	\Repeat{$\Psi({\boldsymbol \xi}^{i},{\boldsymbol \upsilon}^{i},{\boldsymbol \rho}^{\star})-\Psi({\boldsymbol \xi}^{i+1},{\boldsymbol \upsilon}^{i+1},{\boldsymbol \rho}^\star)\le \epsilon$ or $i\ge I_{\rm max}$}
	{
		\nonl \% Allocate computational resources and place the tasks\\
		Given ${\boldsymbol \upsilon}^{i}$ and ${\boldsymbol \xi}^{i}$, solve \eqref{op_c_disjoint} via Algorithm \ref{heuristic_op} and return ${\boldsymbol \upsilon}^{i+1}$ and ${\boldsymbol \xi}^{i+1}$\\
	}
	\KwOutput{${\boldsymbol \xi}^{\rm \star}, {\boldsymbol \upsilon}^{\rm \star}$}
	\caption{DTO Computational Resource Allocation and Task Placement.}
	\label{task_offloading_alg_disjoint}
\end{algorithm}
\subsection{CC of DTO}
\textcolor{black}{CC of DTO is equal to CC of Algorithms \ref{pow_alg_disjoint} and \ref{task_offloading_alg_disjoint}. Algorithm \ref{pow_alg_disjoint} includes two nested loops, wherein the inner loop includes solving \eqref{pow_elas_disjoint} via CVX, whose CC is $\frac{\log (2K+U)}{t^0 \varrho\log \varsigma}$. Moreover, Algorithm \ref{pow_alg_disjoint} minimizes the transmit power via CCP, whose CC is derived in subsection \ref{computational-complexity}. As a result, CC of Algorithm \ref{pow_alg_disjoint} is $CC^6 = (K I_{\rm max}+ I_{\rm max}^{\rho}) \frac{\log (2K+U)}{t^0 \varrho\log \varsigma}$.}

\textcolor{black}{Algorithm \ref{task_offloading_alg_disjoint} includes two nested while loops, in which Algorithm \ref{heuristic} is executed. Based on CC of Algorithm \ref{heuristic} derived in subsection \ref{computational-complexity}, CC of feasibility analysis in Algorithm \ref{task_offloading_alg_disjoint} is $K I_{\rm max}\mathcal{O} (K^2 N B E)$.  Algorithm \ref{task_offloading_alg_disjoint} also includes executing Algorithm \ref{heuristic_op} at most $I_{\rm max}$ times. As a result, CC of Algorithm \ref{task_offloading_alg_disjoint} is $CC^7 = (K+1) I_{\rm max}\mathcal{O} (K^2 N B E)$. Based on above, CC of DTO is $CC^{\rm DTO} = CC^6 + CC^7$. Note that although $CC^{\rm DTO}$ is larger than $CC^{\rm JTO}$, they are in the same order of complexity. The difference of CCs comes from the fact that CCP is performed at most $I_{\rm max}$ times in JTO and only once in DTO.}
\section{Lower Bound on Optimal Solution (LTO)}\label{opato_sec}
Since the optimization problem \eqref{joint_elas} is non-convex, without loss of the optimality, we make some assumptions to resolve the non-convexity of \eqref{joint_elas}. First, we note that it is very likely for the fiber-optic links to have sufficient capacity for carrying the traffic of UEs, which is the case for frontahul links and any wired link in the transport network. As a result, we relax the constraints C3 and C4 from \eqref{joint_elas}. Note that {the relaxation of} C3 and C4 extends the feasible set of \eqref{joint_elas}, resulting in a lower bound on the optimal solution to \eqref{joint_elas}. In addition, with a large number of antenna elements at RRHs, the channel vectors between different RRHs and a specific user are approximately orthogonal, i.e., $|{\bf h}_{u,k}^{\rm H}{\bf h}_{u,j}|\approx 0$ for all $j\neq k$ \cite{lu2014overview}. Therefore, the interference in wireless channels is negligible and \eqref{R_k_first} becomes:
\begin{equation}
\label{R_k_convex}
R_k=W\log_2\left(1+\frac{|{\bf h}_{u,k}|^2}{\sigma^2_{\rm n}}\rho_k\right), \quad k \in \mathcal{K}_u.
\end{equation}
{The elimination of the} interference increases $R_k$ with the same amount of power allocated to each UE, which again results in a lower bound on the optimal solution to \eqref{joint_elas}. Based on the fact that $\mathop{\min}\limits_{{\boldsymbol \alpha}, {\boldsymbol \xi},{\boldsymbol \upsilon},{\boldsymbol \rho}} \sum_{k\in \mathcal{K}}\alpha_k=\mathop{\min}\limits_{{\boldsymbol \alpha}, {\boldsymbol \xi},{\boldsymbol \upsilon}}\left(\mathop{\min}\limits_{\boldsymbol \rho} \sum_{k\in \mathcal{K}}\alpha_k\right)$, the optimal power allocation is the solution to:
\begin{equation}\label{pa}
\begin{array}{ll}
\mathop{\min}\limits_{\boldsymbol \rho}& \sum_{k\in \mathcal{K}}\alpha_k\\
\text{s.t.}&  \text{C1:}\quad T^{\rm exe}_k+T^{\rm prop}_k+T^{\rm tx}_k\le T_k+\alpha_k, \quad \forall k,\\
&\text{C5:}\quad \rho_{k}\le P^{\rm max}_k,\quad \forall k.\\
\end{array}
\end{equation}
The data rate in \eqref{R_k_convex} removes the cross-coupling impact of the allocated power to different users. Hence, without loss of optimality, \eqref{pa} is solved for each $\rho_k$ independently. The associated power allocation problem is:
\begin{equation}\label{pa_k}
\begin{array}{ll}
\mathop{\min}\limits_{{\rho_k}}& T^{\rm tx}_k\\
\text{s.t.}&  \text{C1:}\quad T^{\rm exe}_k+T^{\rm prop}_k+T^{\rm tx}_k\le T_k+\alpha_k, \quad \forall k,\\
&\text{C5:}\quad \rho_{k}\le P^{\rm max}_k,\quad \forall k,\\
\end{array}
\end{equation}
in which $\alpha_k$ in the objective is replaced with $T^{\rm tx}_k$. Note that minimizing $T^{\rm tx}_k$ is equivalent to maximizing $\frac{R_k}{D_k}$. Since $R_k$ in \eqref{R_k_convex} is increasing with $\rho_k$, the optimal solution of \eqref{pa_k} is $\rho_k^\star=P^{\rm max}_k$. Note that feasibility of C1 is ensured by optimizing other variables.\\
\indent Next, we deal with the binary optimization variable $\boldsymbol \xi$. We propose an exhaustive search over all possible values of $\boldsymbol \xi$ to avoid any performance loss due to non-convexity of \eqref{joint_elas}, stemmed from binary $\boldsymbol \xi$. The number of all possible combinations of task placement decisions equals ${|\mathcal{B}|}^{|\mathcal{K}|}$ where $|\mathcal{B}|=\sum_n {|\mathcal{B}_n|}$. Thus, we solve \eqref{joint_elas} for $\boldsymbol \alpha$ and $\boldsymbol \upsilon$ for each task placement decision and select the decision that results in lowest $\sum_k \alpha_k$ as the optimal decision. Note that the exhaustive search may impose an excessive computational complexity. However, LTO is developed as a baseline for performance evaluation and it is not supposed to work in real-time.\\
\indent The optimization problem for solving $\boldsymbol \alpha$ and $\boldsymbol \upsilon$ is:
\begin{equation}\label{kkt_problem}
\begin{array}{ll}
\mathop{\min}\limits_{{\boldsymbol \upsilon},{\boldsymbol \alpha}}& \sum_{k\in \mathcal{K}}\alpha_k\\
\text{s.t.}&  \text{C1-a:}\quad \frac{L_k}{\upsilon_k}\le \tilde{T}_k+\alpha_k, \quad \forall k\in\mathcal{K}\\
&\text{C2:}\quad {\sum_{k\in \mathcal{K}_n}\upsilon_k \le \Upsilon_n, \quad \forall n,}\\
\end{array}
\end{equation}
where $\tilde{T}_k=T_k-T^{\rm prop}_k-T^{\rm tx}_k$ and $\mathcal{K}_n$ is the set of tasks to be executed at executing server $n$. Problem \eqref{kkt_problem} is convex in both $\boldsymbol \alpha$ and $\boldsymbol \upsilon$. As a result, the KKT conditions determine the optimal solution. To derive the KKT conditions, we first write the Lagrangian function as follows:
\begin{equation}\label{lagrange}
\mathcal{L}=\sum_{k\in \mathcal{K}}\left(\alpha_k+\gamma_k( \frac{L_k}{\upsilon_k}- \tilde{T}_k-\alpha_k )-\eta_k\alpha_k-\mu_k\upsilon_k\right)+\sum_{n\in \mathcal{N}}\lambda_n\left(\sum_{k\in \mathcal{K}_n}\upsilon_k -\Upsilon_n\right).
\end{equation}
By derivating the Lagrangian function with respect to $\alpha_k$ and $\upsilon_k$ we have:
\begin{equation}\label{kkt1}
\frac{\partial\mathcal{L}}{\partial\alpha_k}=1-\gamma_k-\eta_k=0, \quad \forall k\in \mathcal{K},
\end{equation}
and
\begin{equation}\label{kkt2}
\frac{\partial\mathcal{L}}{\partial\upsilon_k}=-\gamma_k \frac{L_k}{\upsilon_k^2}-\mu_k+\lambda_n=0, \quad \forall k\in \mathcal{K}_n.
\end{equation}
In addition, the complementary slackness conditions are:
\begin{align}
\gamma_k( \frac{L_k}{\upsilon_k}- \tilde{T}_k-\alpha_k )=0, \quad \forall k\in \mathcal{K},\label{kkt3}\\
\lambda_n\left(\sum_{k\in \mathcal{K}_n}\upsilon_k -\Upsilon_n\right)=0, \quad \forall n\in \mathcal{N},\label{kkt4}\\
\eta_k\alpha_k=0, \quad \forall k\in \mathcal{K},\label{kkt5}\\
\mu_k\upsilon_k=0, \quad \forall k\in \mathcal{K}\label{kkt6}.
\end{align}
Constraint C1-a implies $\upsilon_k>0$. Hence, from \eqref{kkt6} we have $\mu_k=0$ and condition \eqref{kkt2} results in:
\begin{equation}\label{upsilon_kkt}
\upsilon_k=\sqrt{\frac{L_k}{\lambda_n}}, \quad \forall k\in \mathcal{K}_n,
\end{equation}
which implies $\lambda_n>0$. Thus, \eqref{kkt4} gives:
\begin{equation}\label{capacity_kkt}
\sum_{k\in \mathcal{K}_n}\upsilon_k =\Upsilon_n, \quad \forall n\in \mathcal{N}.
\end{equation}
On the other hand, when \eqref{joint_op} is infeasible, we get $\alpha_k>0$. Thus, \eqref{kkt5} leads to $\eta_k=0$ and condition \eqref{kkt1} results in $\gamma_k=1$. As a result, from \eqref{kkt3} we get:
\begin{equation}\label{alpha_kkt}
\alpha_k=\frac{L_k}{\upsilon_k}- \tilde{T}_k, \quad \forall k\in \mathcal{K}.
\end{equation}
Having $\alpha_k\ge0$ and \eqref{upsilon_kkt}, the optimal non-negative variable is:
\begin{equation}\label{alpha_kkt1}
\alpha_k=[\sqrt{L_k \lambda_n}- \tilde{T}_k]^+, \quad \forall k\in \mathcal{K}_n,
\end{equation}
wherein $\lambda_n$ is found such that:
\begin{equation}\label{capacity_eq}
\sum_{k\in \mathcal{K}_n}\frac{L_k}{\tilde{T}_k + \alpha_k} =\Upsilon_n, \quad \forall n\in \mathcal{N}.
\end{equation}
Then, the optimal values of $\alpha_k$ and $\upsilon_k$ are found as in \eqref{alpha_kkt1} and \eqref{upsilon_kkt}, respectively. Having the optimal solution of \eqref{kkt_problem} for all possible $\boldsymbol \xi$, the optimal solution of \eqref{joint_elas} is the solution with lowest objective of \eqref{kkt_problem}.
\section{Simulation Results} \label{simulation-results}
In this section, we evaluate the performance of JTO. The setup of the simulation is presented in Table \ref{simulation_setup}. We assume that $U=4$ RRHs are placed with inter-site distance of 100 m and all users are served in an area of 100 m radius with a given user-RRH assignment. The nodes in the transport network are divided into three tiers based on their distance from UEs: the local tier, the regional tier, and the national tier. Although the number of serving nodes is very large, there are some distant nodes in each tier that impose a large propagation latency. Hence, we only incorporate the nodes with reasonable propagation latency in the transport network \cite{tran2019joint}. Network graph $G$ consists of $N=6$ nodes: $\bar{n}$ at the local tier with zero propagation latency, three nodes at the regional tier with relatively low propagation latency, and two distant nodes at the national tier. For simplicity of comparison, we assume that all nodes have the same computational capacity and all tasks are of the same size, load, and maximum acceptable latency, i.e., $D_k = D$, $L_k = L$, and $T_k = T$, $\forall k$. Moreover, we assume equal propagation latency and capacity for the network links. Note that the relatively low value of link capacity ($0.4$ Gbps) is the amount of capacity solely reserved for task offloading. Finally, the simulations are performed on a 3.30 GHz Core i5 CPU and 16 GB RAM.

\begin{table}[!t]
	\scriptsize
	\renewcommand{\arraystretch}{1.05}
	\centering
	\caption{Simulation Setup.}
	\label{simulation_setup}
	\begin{tabular}{| c| c| c| c| }	
		\hline
		\textbf{Parameter}& \textbf{Value} & \textbf{Parameter} & \textbf{Value}\\\hline		
		$L_k$& $10^6$ CPU Cycles & $\delta_{(m,m')}$& $10$ ms\\\hline
		$M$& $32$ Antennas & $\Lambda_n$& $10^{-28}$ \cite{zhou2018computation}\\\hline
		$D_k$& $0.1$ Mbits & Path Loss& $128.1+ 37.6 \log Q$ \cite{xia2018power}\\\hline		
		$\Upsilon_n$& $10^9$ CPU Cycles per Second{\cite{yang2018mobile}} & $U$& $4$\\\hline
		{$P^{\rm max}_k$}& {$0.5$ Watt} & ISD& $100$ m \\\hline
		$B_{(m,m')}$& $0.4$ Gbps & $W$& $20$ MHz \cite{xia2018power}\\\hline
		$B_{{\rm f},u}$& $0.6$ Gbps & Noise power& $-150$ dBm/Hz \cite{xia2018power}\\\hline
	\end{tabular}
\end{table}
{Fig. \ref{ar} (a)} reports the performance of the feasibility analysis in JTO, showing the acceptance ratio versus $T$. The acceptance ratio is defined as the ratio of accepted services by the feasibility analysis over the total number of the requested tasks. Note that the acceptance ratio increases by increasing $T$. This is due to the fact that the tasks with higher $T$ need less transmit power and computational resources to be served. Moreover, for higher $T$, a larger number of nodes are available for task offloading.
\textcolor{black}{In addition, we solve \eqref{joint_elas} by the alternate search method (ASM), in which \eqref{joint_elas} is alternatively solved for each variable. Note that the sub-problem of $\boldsymbol \upsilon$  is solved by CVX and the sub-problem of $\boldsymbol \xi$  is solved by MOSEK (details are not provided due to space limitation). The effectiveness of JTO against ASM is also shown in {Fig. \ref{ar} (a)}. Note that for latencies smaller than $75$ ms, JTO outperforms ASM. Moreover, the performance of both methods is identical for low values of $ T$. This is due to the fact that the set of accessible NFV-enabled nodes for low values of $T$ is restricted to $\bar{n}$ and therefore, JTO is not able to offload the tasks to more distant NFV-enabled nodes because their propagation latencies violate the E2E latency constraints.}

The acceptance ratio of JTO for different number of tasks is shown in {Fig. \ref{ar} (b)}. Since the amount of available resources is limited, the acceptance ratio is decreasing with {the} increase in the total number of {tasks}. Again, the superiority of JTO over ASM is observed.
    \begin{figure}[h]
        \centering
        \subfloat[Acceptance ratio vs. $T$ for $K=30$.]{	\includegraphics[width=7cm, height=3.5cm]{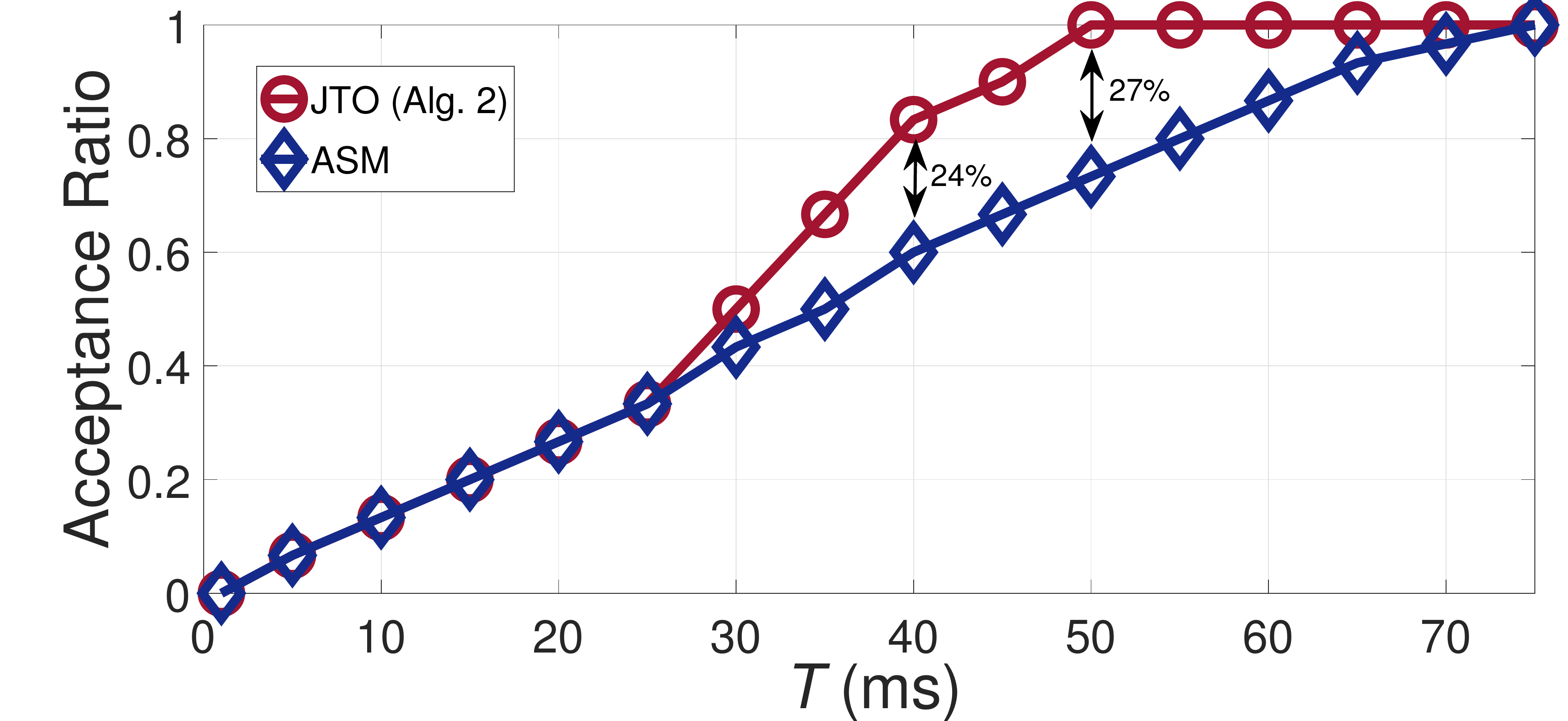}\label{ar_vs_tau}}
        \subfloat[Acceptance ratio vs. $K$ for $T=40$ ms.]{ \includegraphics[width=7cm, height=3.5cm]{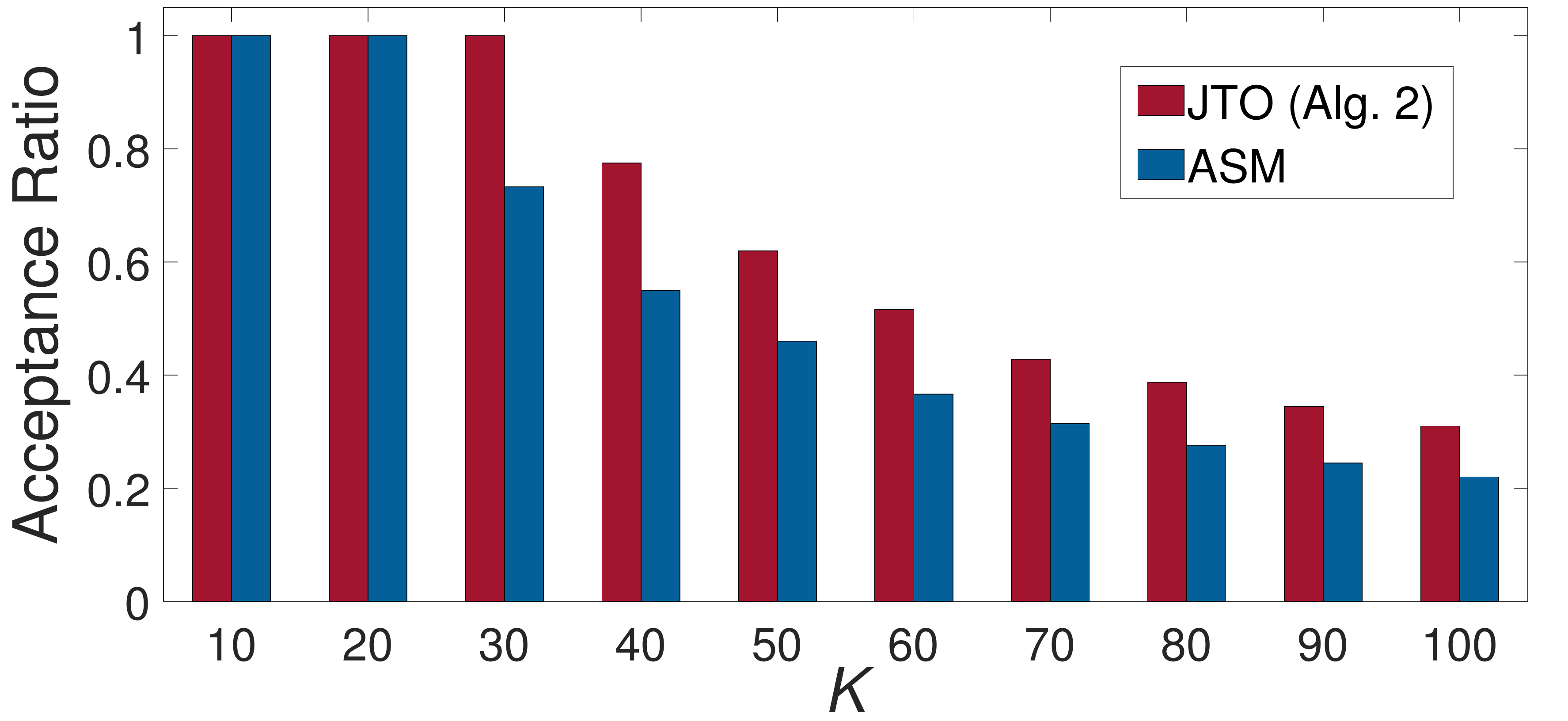}\label{ar_vs_k}}
        \caption{\small Acceptance ratio vs. $T$ and $K$}\label{ar}
    \end{figure}

\vspace{-20 pt}
The convergence of Algorithm \ref{feasiblity_alg} is shown in {Fig. \ref{elas_fig} (a)}. As {expected}, the sum of non-negative variables is decreasing in each iteration. Furthermore, Algorithm \ref{feasiblity_alg} converges faster than ASM, which stems from higher acceptance ratio of JTO.
    \begin{figure}[h]
        \centering
        \subfloat[Convergence of admission control algorithm for \newline $ T=20$ ms and $K=30$.]{%
	\label{elas_vs_iteration}
	\includegraphics[width=7 cm, height=4 cm]{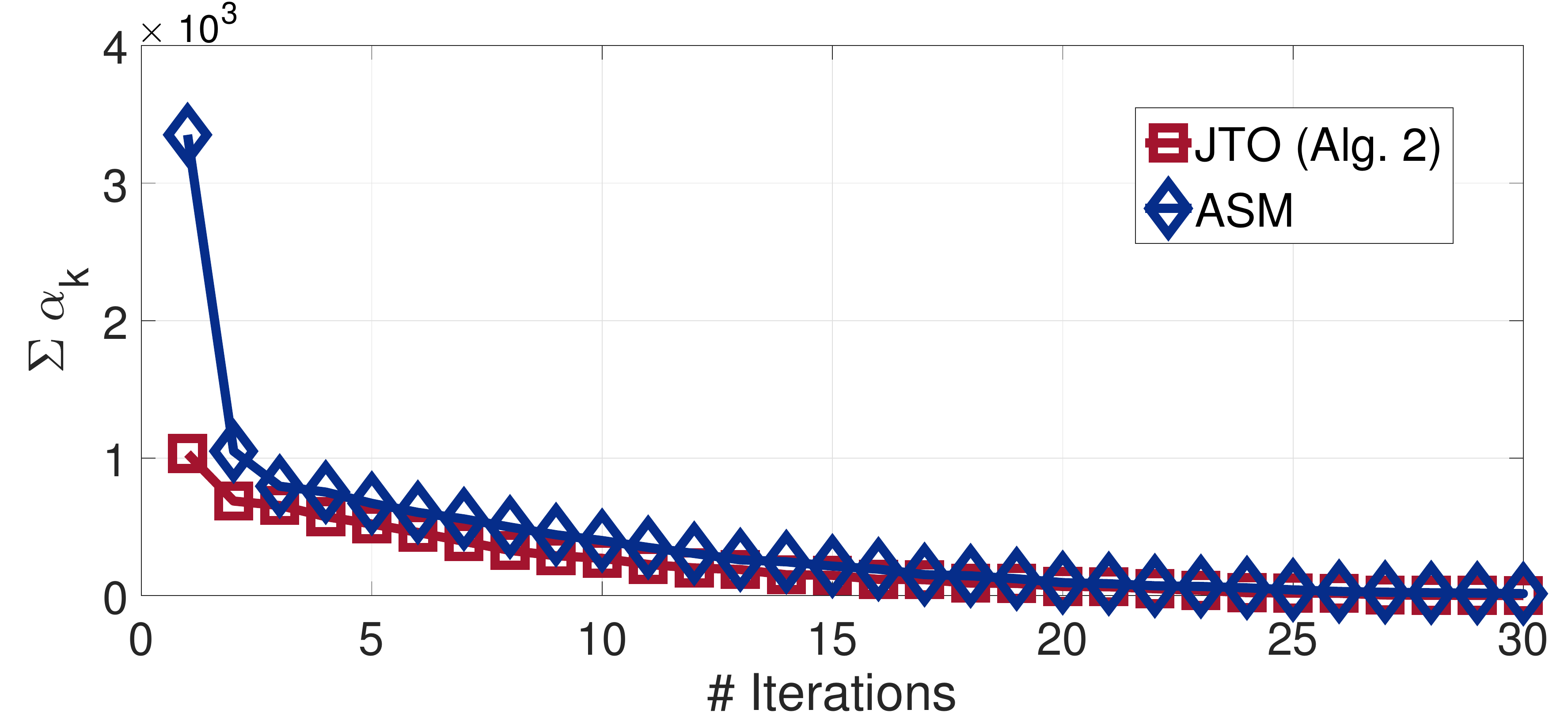}}
        \subfloat[Acceptance ratio of joint vs. disjoint methods in terms of $T^{\rm RAN}$ for $ T=30 $ ms and $K=30$ users. ]{%
           \includegraphics[width=7 cm, height=4 cm]{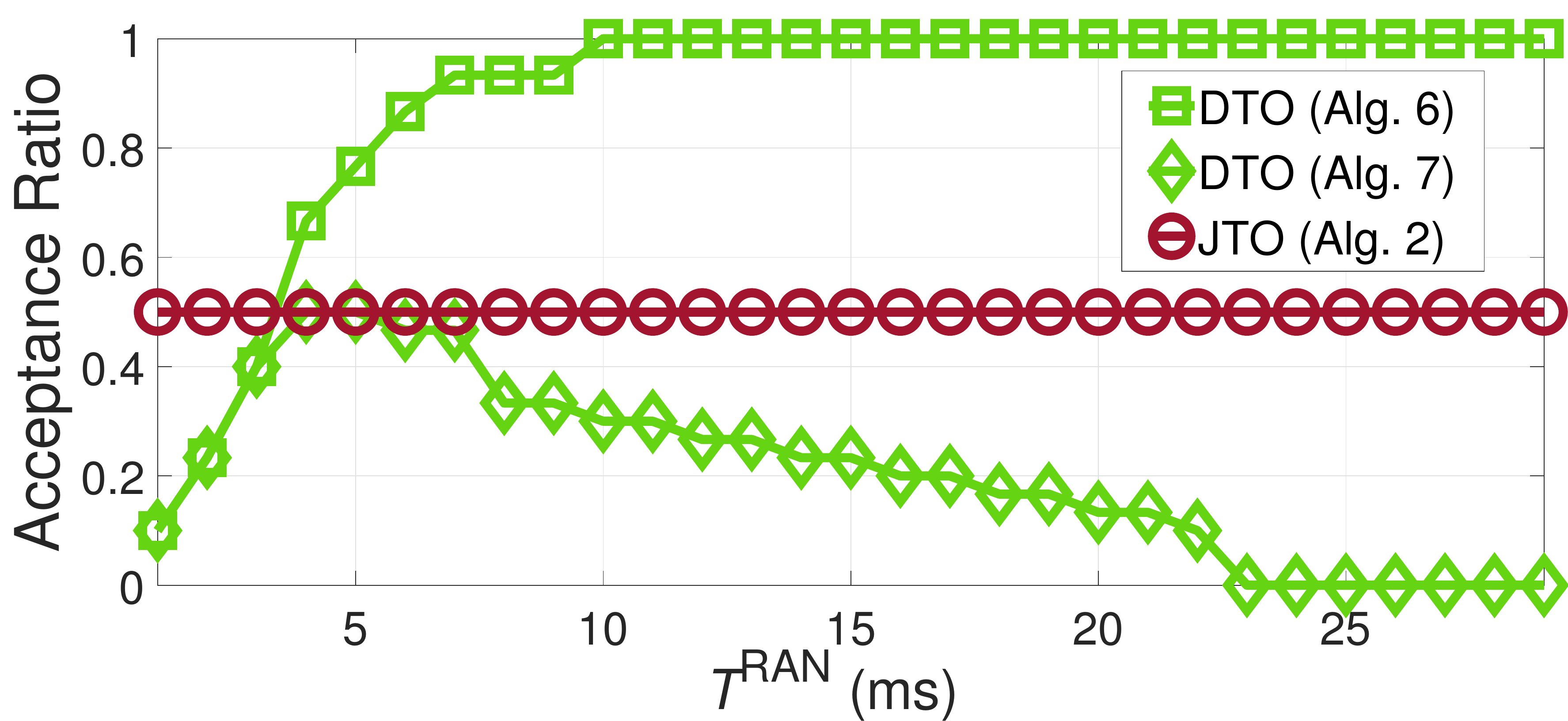}\label{joint_vs_disjoint}}
        \caption{\small Convergence and acceptance ratio of the proposed methods.}
        \label{elas_fig}
    \end{figure}

\textcolor{black}{The acceptance ratio of JTO is compared with DTO in Fig. \ref{elas_fig} (b). The acceptance ratio of JTO and DTO is depicted for $ T=30$ ms. For DTO, we obtain the acceptance ratio for different values of $T^{\rm RAN}\in (0, T)$. Moreover, the acceptance ratio of the feasibility analysis in the transmit power allocation phase of DTO, i.e., Algorithm \ref{pow_alg_disjoint}, is depicted. The acceptance ratio of DTO is increasing for small values of $T^{\rm RAN}$, that is, the small values of $T^{\rm RAN}$ impose high rates on users, which is impossible due to either insufficient bandwidth or limited fronthaul capacity. On the other hand, for larger values of $T^{\rm RAN}$, the acceptance ratio of Algorithm \ref{pow_alg_disjoint} is 1 but the task placement and computational resource allocation restricts the number of accepted tasks. Furthermore, JTO outperforms DTO in different values of $T^{\rm RAN}$.}

\textcolor{black}{{Fig. \ref{avg_fig} (a)} shows the average radio transmission latency, i.e., $\frac{1}{K}\sum_{k\in \mathcal{K}}T^{\rm tx}_k$, and the average execution latency of tasks, i.e., $\frac{1}{K}\sum_{k\in \mathcal{K}}T^{\rm exe}_k$ for different values of $D$ given $T=20$ ms. The average radio transmission latency increases by increasing $D$ and subsequently the average execution latency is decreased to maintain the maximum acceptable latency. Therefore, it is inferred that JTO efficiently manages the transmit power and the computational resources. Similarly, according to {Fig. \ref{avg_fig} (b)}, the average execution latency increases by increasing $L$ and subsequently this increase is compensated with lower radio transmission latency.}
    \begin{figure}[h]
        \centering
        \subfloat[Average radio transmission and execution latencies vs. $D$ for $T=20$ ms and $K=30$.]{%
	\includegraphics[width=7.5cm]{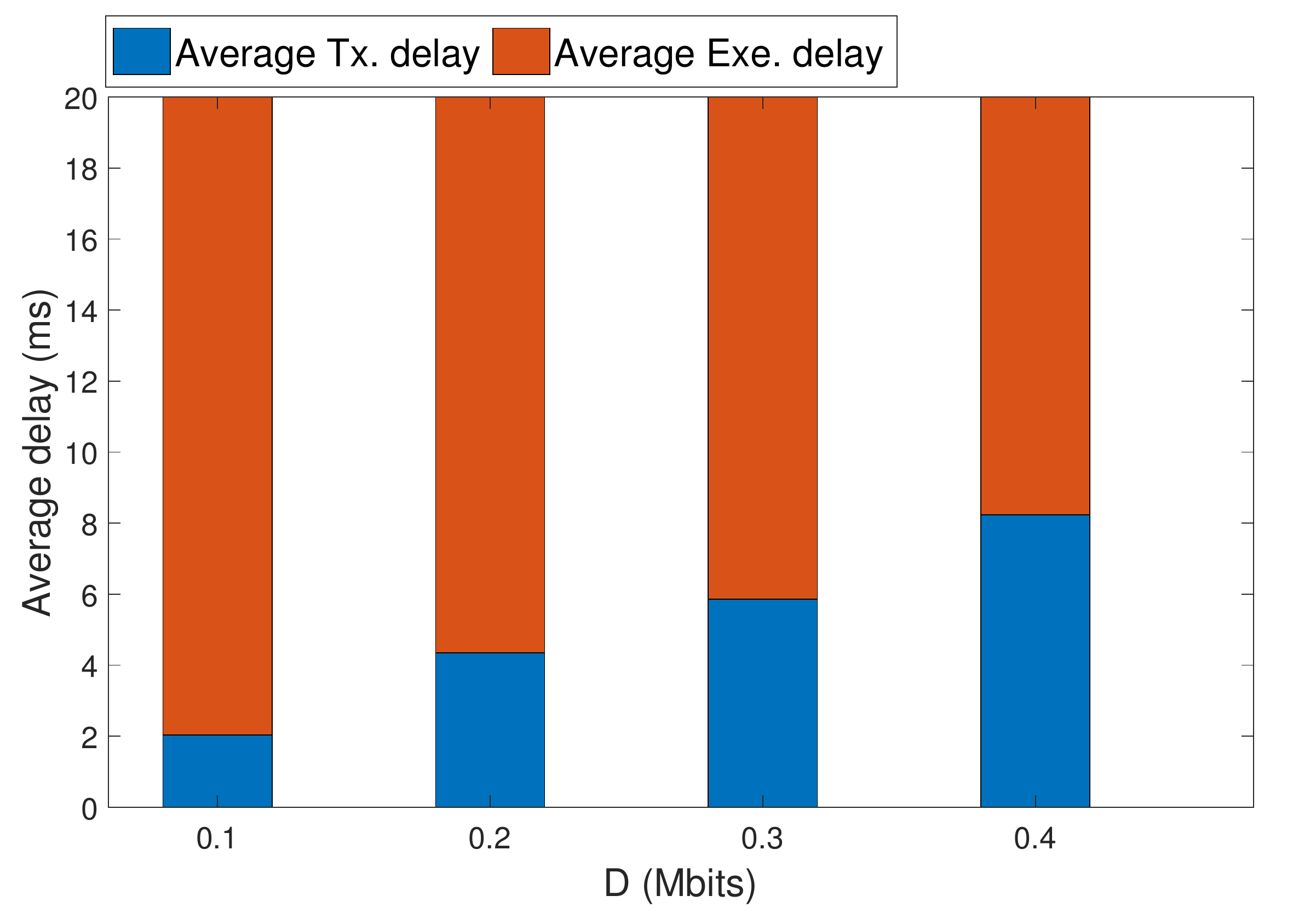}\label{Avg_delay_vs_D}}
\quad
        \subfloat[Average radio transmission and execution latencies vs. $L$ for $ T=20 $ ms and $K=30$. ]{
            \includegraphics[width=7.5cm]{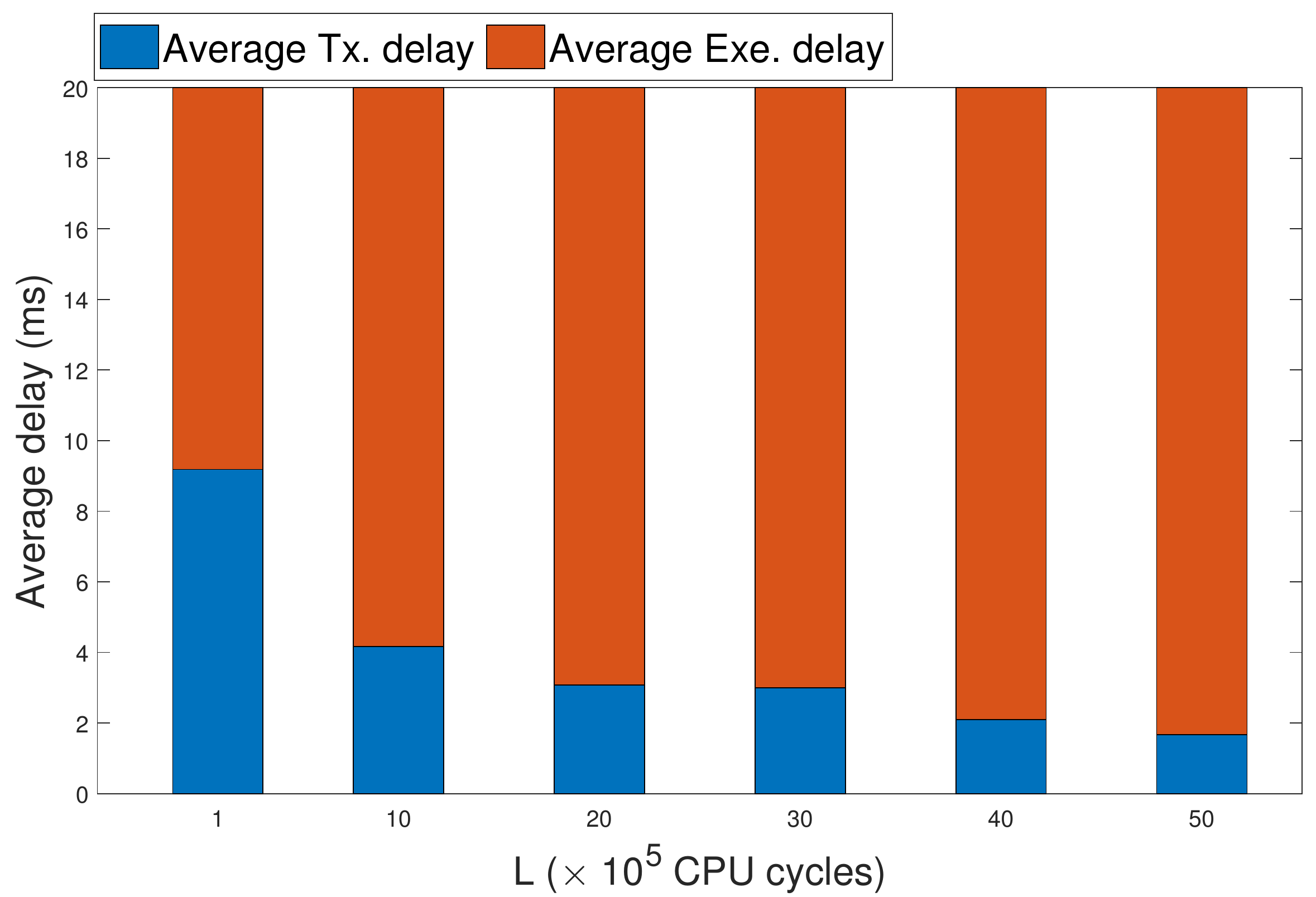}\label{Avg_delay_vs_L}}
        \caption{\small Average radio transmission and execution latencies in JTO.}
        \label{avg_fig}
    \end{figure}
\vspace{-20 pt}

\textcolor{black}{In Fig. \ref{tasks_nodes}, we assume there are three classes of tasks (each including $10$ tasks) with three different maximum acceptable latencies, i.e., $ T^{(1)}=10 $ ms, $ T^{(2)}$ =50 ms, and $ T^{(3)}=100$ ms. The classes $(1)$, $(2)$, and $(3)$ are considered as the sets of tasks with low, medium, and high latency requirements, respectively. Moreover, we assume there are three nodes (shown by rectangles): a local node (i.e., $\bar n$) with zero propagation latency, a regional node with $20$ ms propagation latency, and a national node with $40$ ms propagation latency. The propagation latencies are the summation of uplink and downlink propagation latencies. Fig. \ref{tasks_nodes} shows the task placement for different values of the processing capacity of nodes $C=\Upsilon_n,\forall n$. When $C=1$, none of the nodes is able to serve the tasks in class (1) due to their high resource utilization. However, the tasks in class $(2)$ are mainly served at the local node and class $(3)$ tasks are placed at regional and national nodes. When $C=10$, some of the tasks in class (1) are placed at the local node. Moreover, some tasks in class (2) and (3) are served at the local node as well. Furthermore, the national node does not serve any task because JTO places the tasks at the nearest nodes in order to reduce the transmit power. When $C=20$, more tasks in class (1) are served at the local node and the acceptance ratio reaches 1. Finally, when $C=50$, almost all of the tasks are placed at the local node to reduce the transmit power consumption.}
	\begin{figure}[!t]
		\centering
		\includegraphics[width=9cm]{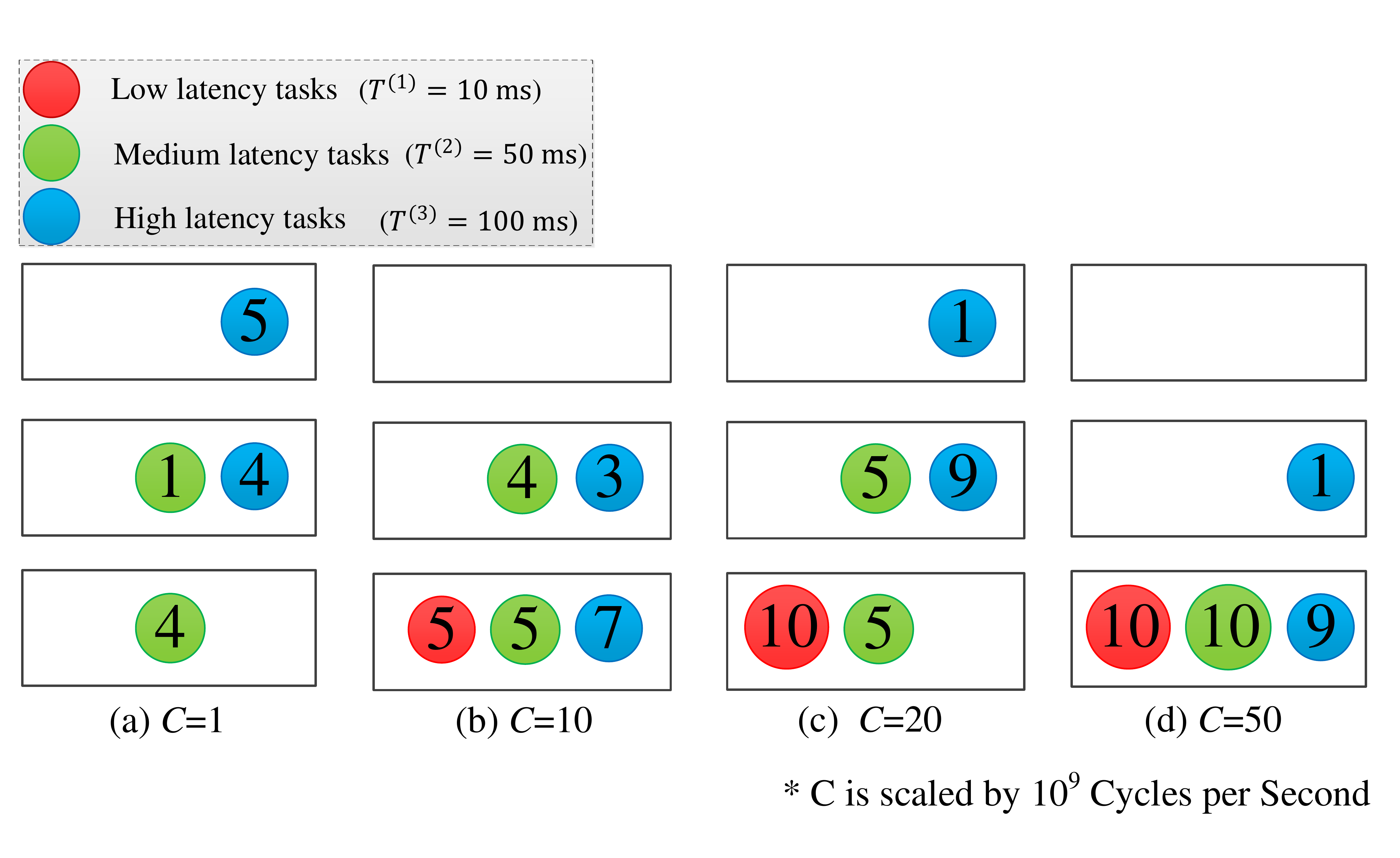}
		\caption{\small Placement of the different classes of tasks at three different tiers of nodes for $K=30$.}\label{tasks_nodes}
	\end{figure}
		
\textcolor{black}{Table \ref{ar_table} shows the acceptance ratio of each class for different values $C$. Note that the acceptance ratio of all classes is increased by increasing $C$. Moreover, the acceptance ratio of class (1) is lower than that of classes (2) and (3). The reason is twofold, one is due to high resource utilization by tasks of this class and another is due to limited number of available nodes for tasks with low latency requirement (only node $\bar n$ in this example).}
\begin{table}[!h]
	\renewcommand{\arraystretch}{1}
	\centering
	\caption{\small Acceptance ratio of {JTO for} different task classes vs. processing capacity of nodes.}
	\label{ar_table}
\begin{tabular}{cccc}
	\hline
	\textbf{\begin{tabular}[c]{@{}c@{}}Computational\\ capacity {($10^9$ CPU cycles/sec)} \end{tabular}} & \multicolumn{3}{c}{\textbf{Maximum acceptable latency (ms)}} \\ \cline{2-4} 
	& $ T^{(1)}=10$ & $ T^{(2)}=50$ & $ T^{(3)}=100$ \\ \hline
	$C=1$ & $0$ & $0.5$ & $0.9$ \\ 
	$C=10$ & $0.5$ & $0.9$ & $1$ \\
	$C=20$ & $1$ & $1$ & $1$ \\ \hline
\end{tabular}
\end{table}

\textcolor{black}{Fig. \ref{opato} shows the acceptance ratio of LTO and JTO for different values of maximum acceptable latency $T$. Due to the high computational complexity of exhaustive search in LTO, we consider a simple network graph comprised of two nodes connected with a single link. Moreover, the total number of tasks $|\mathcal{K}|$ is 20. The acceptance ratio of both JTO and LTO is lower for larger computational loads. Meanwhile, the acceptance ratio of JTO is almost the same as LTO for different values of $T$ and $L$.}
	\begin{figure}[!h]
		\centering
		\includegraphics[width=9cm]{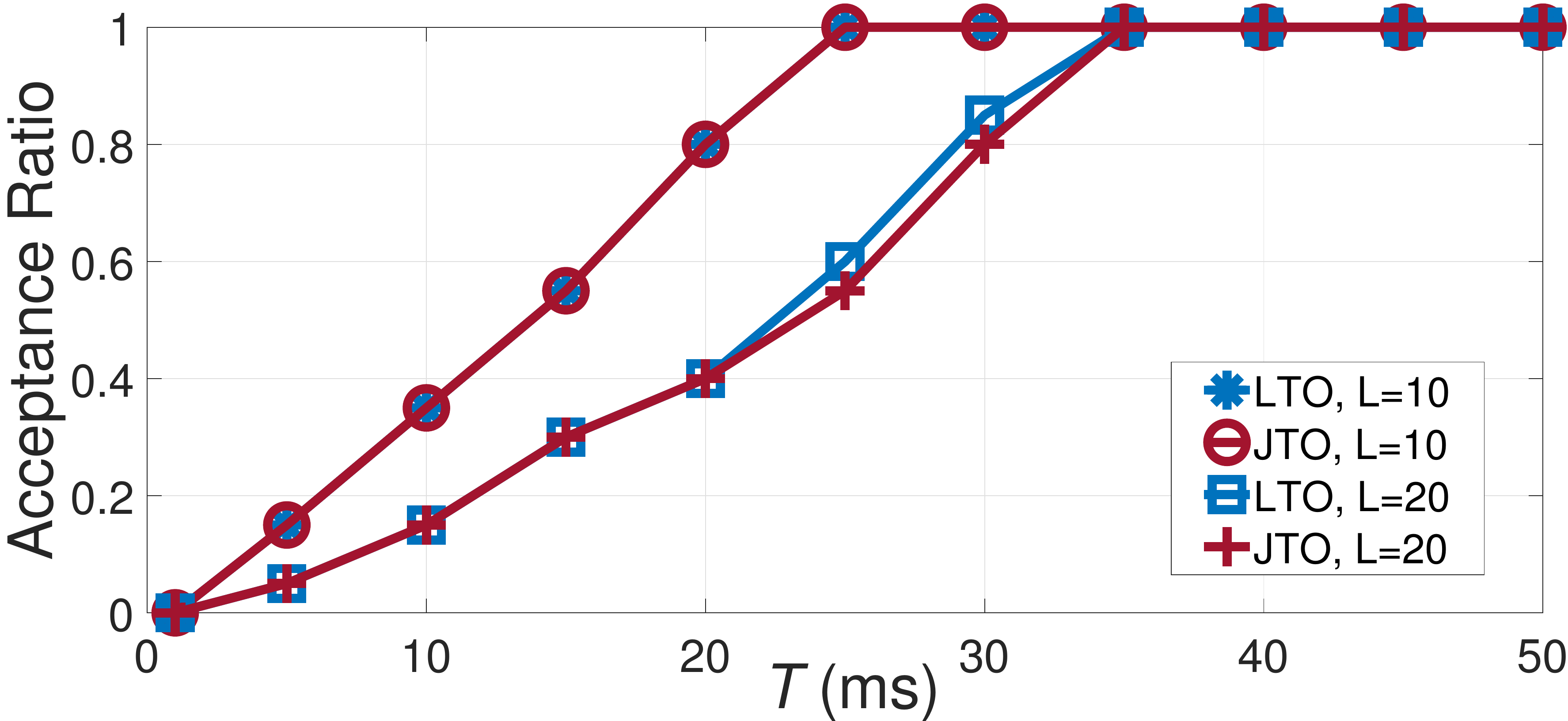}
		\caption{\small Acceptance ratio of LTO and JTO {vs.}  maximum acceptable latency.}\label{opato}
	\end{figure}
\vspace{-20 pt}
\section{Conclusions and Future Work} \label{conclusion}
\textcolor{black}{In this paper, we considered an energy-efficient task offloading problem under E2E latency constraints. We investigated the joint impact of radio transmission, propagation of tasks through the transport network, and execution of tasks on the experienced latency of tasks. Due to the non-convexity of the optimization problem, we decoupled the transmit power allocation from task placement and computational resource allocation. {The transmit power allocation was solved by adopting CCP} to convexify the sub-problem. The task placement and computational resource allocation were solved via our proposed heuristic method, which minimizes the sum of propagation and execution latencies. Furthermore, to ensure the feasibility of the optimization problem, we proposed a feasibility analysis that eliminates the tasks causing infeasibility. Simulation results showed the superiority of JTO over both DTO and ASM. The performance of DTO depended on the part of latency required to be met in the radio access network, i.e., $T^{\rm RAN}$. However, JTO showed higher acceptance ratios for different values of $T^{\rm RAN}$. As future work, we plan to incorporate task scheduling into JTO. Moreover, the investigation of an innovative solution that divides the required computational load of each task among several nodes will be an interesting future research activity.}

\bibliographystyle{ieeetr}
\bibliography{citation_Offloading}{}

%
%
%
%
%




\end{document}